\documentclass[11pt]{amsart}
\usepackage[utf8]{inputenc}
\usepackage[T1]{fontenc}
\usepackage{amssymb}
\usepackage{amsthm}
\usepackage{thmtools}
\usepackage{graphics}
\usepackage{xcolor}
\usepackage{enumitem}
\usepackage[pagebackref, colorlinks = true, linkcolor = teal, urlcolor  = teal, citecolor = violet]{hyperref}
\usepackage[margin=1in]{geometry}
\usepackage{cleveref}
\crefname{equation}{Eq.}{Eqs.}
\usepackage{bbm}

\usepackage{scalerel,stackengine}
\stackMath
\renewcommand\widehat[1]{%
	\savestack{\tmpbox}{\stretchto{%
			\scaleto{%
				\scalerel*[\widthof{\ensuremath{#1}}]{\kern-.6pt\bigwedge\kern-.6pt}%
				{\rule[-\textheight/2]{1ex}{\textheight}}%
			}{\textheight}%
		}{0.5ex}}%
	\stackon[1pt]{#1}{\tmpbox}%
}

\newtheorem{theorem}{Theorem}[section]
\newtheorem{definition}[theorem]{Definition}
\newtheorem{proposition}[theorem]{Proposition}
\newtheorem{corollary}[theorem]{Corollary}
\newtheorem{lemma}[theorem]{Lemma}
\newtheorem{remark}[theorem]{Remark}

\newcommand{\id}{{\rm Id}}

\newcommand{\one}{{\mathbf 1}}
\newcommand{\tr}{{\operatorname{tr}}}

\renewcommand{\d}{{\rm d}}

\newcommand{\pp}{{\mathbb P}}
\newcommand{\ee}{{\mathbb E}}

\newcommand{\nn}{{\mathbb N}}
\newcommand{\cc}{{\mathbb C}}

\newcommand{\supp}{\operatorname{supp}}

\newcommand{\as}{\operatorname{-}\mathrm{a.s.}}

\title[Selection of sectors for quantum trajectories]{Exponentially fast selection of sectors for quantum trajectories beyond non-demolition measurements}

\author[Benoist]{Tristan Benoist}
\address{Institut de Math\'ematiques de Toulouse, UMR5219, Universit\'e de Toulouse, CNRS, UPS IMT, F-31062 Toulouse Cedex 9, France}
\email{tristan.benoist@math.univ-toulouse.fr}

\author[Greggio]{Linda Greggio}
\address{Laboratoire de Physique de l’Ecole Normale Supérieure, Mines Paris, Inria, CNRS, ENS-PSL, Sorbonne Université, PSL Research University, Paris, France}
\email{linda.greggio@inria.fr}

\author[Pellegrini]{Clément Pellegrini}
\address{Institut de Math\'ematiques de Toulouse, UMR5219, Universit\'e de Toulouse, CNRS, UPS IMT, F-31062 Toulouse Cedex 9, France}
\email{clement.pellegrini@math.univ-toulouse.fr}

\date{\today}

\begin{document}
\begin{abstract}
    We show that, in long time, quantum trajectories select an invariant subspace of the Hilbert space of the system being indirectly measured. This selection is shown to be exponentially fast in an almost sure sense and in average. Moreover, the selection mimics a unique positive operator measurement. This generalizes known results for non-demolition measurements to arbitrary repeated indirect measurements. Our proofs are based on the introduction of a deformation of the original instrument to an equivalent one restricted to some subspace.
\end{abstract}

\maketitle
\tableofcontents
\section{Introduction}

Quantum trajectories describe the evolution of a quantum system undergoing repeated indirect measurements \cite{braginsky1995quantum, haroche2006exploring}. %
They are fundamental modeling tools in the theory of measurement and control of quantum systems, especially for quantum optics experiments -- see \cite{breuer2002theory,gardiner2004quantum,haroche2006exploring,wiseman2009quantum}. From a mathematical perspective, quantum trajectories are specific Markov chains that take values in the set of the system states. 

In this article we are interested in how quantum trajectories select invariant subspaces in the long time limit. If every invariant subspace is one dimensional and the transient subspace is trivial, the indirect measurement process is called a quantum non-demolition (QND)  measurement. This terminology was introduced by physicists in \cite{braginsky1995quantum}. In continuous time, the resulting process ressembles the stochastic collapse models introduced by Gisin~\cite{gisin} and Diosi~\cite{diosi}. In \cite{adler}, Adler \emph{et al.} used a martingale approach to prove the collapse in these models. From an experimental point of view, an important milestone was the QND measurement of a number of photon in a superconducting cavity by S. Haroche's group \cite{guerlin_progressive_2007}. That inspired  Bauer and Bernard to prove, not only the reproduction of the wave function collapse postulate, but also an exponential rate for the collapse~\cite{BenrardBauer} -- see also \cite{BenoistAHP}. In \cite{amini2023exponential} the authors generalized the equivalence to wave function collapse when the invariant subspaces are of arbitrary dimension. They also prove a different kind of convergence rate for the collapse. However, they assumed that the direct sum of all the invariant subspaces was equal to the whole space, namely that the system did not present any transient subspace. The escape from a transient subspace has been studied, for the average evolution in \emph{e.g.} \cite{Ticozzi2009}. The almost sure convergence with exponential rate was proved in \cite{BenoistTicco}. In the present article we complete this picture, proving exponentially fast selection of subspaces, mimicking wave function collapse, without any prior assumptions on the shape of the invariant subspaces. For that purpose, denoting $\mathcal{H}$ the Hilbert space of the system, we show there exist a unique positive operator valued measure (POVM) and a related orthogonal partition $\{\mathcal{K}_\alpha\}_\alpha$ of the recurrent subspace of $\mathcal{H}$ such that the system state concentrates exponentially fast onto one of these subspace with a law given by the POVM. We call the subspaces forming the orthogonal partition sectors, in accordance with \cite{BenoistAHP}.

We emphasize that our result provides a full generalization of previous results regarding the selection of invariant spaces for quantum trajectories. The sectors and the related POVM are not fixed a priory but are intrinsic to the instrument used to measure. The results of \cite{BenoistAHP, BenrardBauer, Amini20211} become corollaries to \Cref{thm:generalized QND,thm:meanconv}. They are related to the example of \Cref{sec:no transient}, where no transient subspace is present. It is worth noting that \cite{Amini20211} did not consider the almost sure convergence rate, and its identifiability condition inherently forbids the presence of non-trivial sectors. It is our introduction of sectors and the use of a deformation of quantum channels restrictions that allows for a generalization to all possible scenarios for quantum channels and instruments.

The sectors definition follows from the decomposition of shift invariant measures over measurement outcomes into ergodic components. Then, the construction of the related subspaces follows from the decomposition of the Hilbert space into invariant subspaces. The existence and classification of invariant subspaces were studied under the name "enclosures" in \cite{Baumgartner2012structures,carbone2016irreducible}. The sectors we define are direct sums of minimal enclosures. It provides a classification of enclosures depending on the statistics of the measurement outcomes they induce. The POVM that determines the law of the selected sector is constructed using the absorption operators introduced in \cite{carbone2021absorption}. 

Based on the minimal enclosure decomposition, in \cite{carbone2022generalized}, the authors have extended several limit theorems for the statistics of measurement outcomes to situations when multiple minimal enclosures exist. In \cite{BePelLin} we review and extend these results to other limit theorems and some concentration inequalities using our definition of sectors. Here, we focus on the system evolution conditioned on the measurement outcomes. We demonstrate that the sector selection occurs exponentially fast, both almost surely and on average. 

We prove the convergence of each quantum trajectory to a random sector using a standard martingale convergence argument, similar to those developed in \cite{adler, BenrardBauer, BenoistAHP, BenoistTicco}.

The almost sure exponential convergence rate is derived in terms of the relative entropy between sectors. The proof relies on tools from ergodic theory, such as Kingman's ergodic theorem, inspired by \cite{Benoistentrop1}. However, here, the possible absence of a full rank invariant state implies the outcome law is not always dominated by the invariant state one. Hence, we develop a finer analysis of some limits under the left-shift. 

For the average exponential rate, we use a Lyapunov approach similar to the one developed in \cite{amini2023exponential}. Notably, we extend the results of this reference, taking into account that a full-rank invariant state may not exist and dealing with more general measurement procedures, including imperfect measurements.

On top of the convergence for quantum trajectories, we show the same convergence and bounds hold for an appropriate filter of them. Filtering is an important aspect of quantum estimation theory. The filter problem addresses the estimation of quantum trajectories when the initial state is unknown. An estimated trajectory (starting with an arbitrary initial state) is updated using only the measurement outcomes obtained from the true trajectory. We show that such a filter has the same behavior as the true trajectory, that is exponential selection of the same sector with the same rates. This is reminiscent of results on the stability of filters -- see \cite{BenoistAHP,Rouchon2011,Amini2014,Amini20211}.

Our proofs are based on the construction of instruments restricted to sectors. These instruments are deformations of the reference one and allow to write the outcome laws as mixtures of outcome laws verifying ergodic properties. Then, these ergodic properties are transferred back to the original outcome law in the spirit of \cite{BenoistAHP,BenoistJPHYSA,benoisGamb}.
\bigskip

The article is structured as follows. In Section~\ref{sec:setup}, we define instruments, invariant states and effects, the associated outcomes laws and quantum trajectories. In particular we define the notion of sectors through a decomposition of shift invariant outcome laws into ergodic component. Then, in \Cref{thm:sectors POVM}, we relate this definition to invariant subspaces (enclosures) and a POVM. In \Cref{sec:main results} we state our main results about convergence and related exponential rates. In \Cref{sec:proof sectors} we prove the theorems of \Cref{sec:setup} to be able to discuss several examples in \Cref{sec:xpl}. Then, in \Cref{sec:deformed instruments} we introduce the deformed instruments that allows us to reduce to ergodic outcome laws. In Sections~\ref{sec:as} and \ref{sec:mean} we present the proofs of respectively the almost sure and mean rate of sector selection.

\section{Repeated measurements and sectors}\label{sec:setup}
\subsection{Outcomes law}
The law of quantum repeated measurement outcomes is described by an instrument $\mathcal J$. We assume that the set of possible measurement outcomes $\mathcal A$ is finite. We restrict ourselves to finite dimensional Hilbert spaces. The set of linear maps from a Hilbert space $\mathcal H$ (modeling the system) to itself is denoted $\mathcal B(\mathcal H)$. 
Then the instrument is an indexed set of completely positive (CP) maps from $\mathcal B(\mathcal H)$ to itself summing up to an identity preserving CP map or Quantum Channel:
$$\mathcal J:=\{\Phi_a\}_{a\in\mathcal A},\quad \Phi:=\sum_{a\in \mathcal A}\Phi_a\mbox{ s.t. }\Phi(\id_{\mathcal H})=\id_{\mathcal H}.$$
The set of states is the set $\mathcal D(\mathcal H)$ of density operators on $\mathcal H$:
$$\mathcal D(\mathcal H)=\{\rho\in\mathcal B(\mathcal H) : \rho\geq 0, \tr\rho=1\}.$$
In the physics literature, effects are operators allowing to compute the statistics of a yes/no outcome measurement. The operator $E$ is an effect if $0\leq E\leq \id_{\mathcal H}$. Then, $\{E,\id_{\mathcal H}-E\}$ is a two outcome POVM. We denote $\mathcal E(\mathcal H)$ the set of effects:
$$\mathcal E(\mathcal H)=\{E\in \mathcal B(\mathcal H): 0\leq E\leq \id_{\mathcal H}\}.$$
We omit the dependence in $\mathcal H$ when it does not introduce any confusion.

The dual of a map $T:\mathcal B(\mathcal H)\to\mathcal B(\mathcal H)$ with respect to the Hilbert-Schmidt inner product (\emph{i.e.} the Schrödinger picture evolution) is denoted $T^*$. Thus, $\Phi^*$ preserves the trace: $\tr\circ\Phi^*=\tr$.

A sequence of measurement outcomes is an element of $\Omega:=\mathcal A^\nn$. We denote $\Omega_{\rm fin.}:=\cup_{n\in \nn} \mathcal A^n$ the set of finite sequences. For two words $\mathbf{a}=(a_1,\dotsc,a_n)$ and $\mathbf{b}=(b_1,\dotsc,b_p)$, we use the notation $\mathbf{ab}=(a_1,\dotsc,a_n,b_1,\dotsc,b_p)$, for their concatenation.  

For any $n\in\nn$, let $\mathcal F_n$ be the smallest $\sigma$-algebra making measurable the cylinder sets
$$C_{\mathbf{a}}=\{\omega\in \Omega: \omega_1=a_1,\dotsc,\omega_{n}=a_n\}$$
where $\mathbf{a}\in \mathcal A^n$ for an arbitrary $n\in \nn$. The sequence $(\mathcal F_n)_{n\in\nn}$ is a filtration and  we set $\mathcal F$ as the smallest $\sigma$-algebra such that $\mathcal F_n\subset \mathcal F$ for any $n\in \nn$. The $\sigma-$algebra $\mathcal F$ is the so-called \textit{cylinder} $\sigma-$algebra and $(\Omega,\mathcal F,(\mathcal F_n)_{n\in \nn})$ is a filtered measurable space.

Following the postulates of quantum mechanics, given an instrument $\mathcal J$ and an initial system state $\rho\in \mathcal D(\mathcal H)$, the law of the sequence of outcomes is defined through Kolmogorov extension theorem by
$$\pp_\rho(C_{\mathbf{a}})=\pp_\rho(\mathbf{a})=\tr(\rho\Phi_{a_1}\circ\dotsb\circ\Phi_{a_n}(\id))$$
for any $\mathbf{a}=(a_1,\dotsc,a_n)\in \Omega_{\rm fin.}$. This probability is defined on $(\Omega,\mathcal F)$. To lighten the notation, we sometimes use the notation $\Phi_{\textbf{a}}=\Phi_{a_1}\circ\dotsb\circ\Phi_{a_n}$. Furthermore, having observed the first $n$ outcomes to be $\mathbf{a}=(a_1,\dotsc,a_n)$, the state after these $n$ measurements and conditioned on the observation of $\mathbf{a}$ is
\begin{equation}\label{eq:defrhon}   \rho_n=\frac{\Phi^*_{a_n}\circ\ldots\circ\Phi^*_{a_1}(\rho)}{\tr( \Phi^*_{a_n}\circ\ldots\circ\Phi^*_{a_1}(\rho))}=\frac{\Phi^*_{a_n}(\rho_{n-1})}{\tr(\Phi^*_{a_n}(\rho_{n-1}))},\quad\rho_0=\rho.
\end{equation}
Note that the observation of $\mathbf{a}=(a_1,\dotsc,a_n)$ appears with probability 
$\tr( \Phi^*_{a_n}\circ\ldots\circ\Phi^*_{a_1}(\rho))$. Hence, if $\tr( \Phi^*_{a_n}\circ\ldots\circ\Phi^*_{a_1}(\rho))=0$, meaning that the result $\mathbf a$ cannot be observed, the above expression is ill-defined. In that case,  we arbitrarily impose a value $\mu$ to $\rho_n$, where $\mu\in\mathcal D(\mathcal H)$. This construction is fictitious since this arbitrary assignment happens with $\pp_\rho$-probability $0$. Indeed, it is a tautology to say that for all $n\in\mathbb{N}^*$,
\[
\mathbb{P}_\rho(\{\omega : \tr( \Phi^*_{\omega_n}\circ\ldots\circ\Phi^*_{\omega_1}(\rho))=0\})=0.
\]

 Then, given that the law of the outcome sequence is $\pp_\rho$, the process $(\rho_n)$ defined by \eqref{eq:defrhon} is well defined and is a Markov chain called \textit{quantum trajectory.} We are interested in the behavior of this Markov chain as $n$ grows to infinity.

 \subsection{Sectors}
 In this section we define the sectors and related objects.

 Given a quantum channel $\Phi$ we denote $\mathcal D_{\Phi^*}$ the set of $\Phi^*$-invariant states,
 $$\mathcal D_{\Phi^*}:=\{\rho\in \mathcal D: \Phi^*(\rho)=\rho\}.$$
 This set is related to the set of measures $\pp_\rho$ invariant under the left-shift $\theta$ on $\Omega$:
 $$\theta(\omega_1,\omega_2,\dotsc)=(\omega_2,\omega_3,\dotsc).$$
 We denote $\mathcal M(\mathcal J)$ the image of $\mathcal D(\mathcal H)$ by $\rho\mapsto \pp_\rho$. It is a subset of the probability measures on $\Omega$:
 $$\mathcal M(\mathcal J)=\{\pp_\rho:\rho\in \mathcal D(\mathcal H)\}.$$
We denote by $\mathcal M_\theta(\mathcal J)$ the set of elements of $\mathcal M(\mathcal J)$ that are $\theta$-invariant:
$$\mathcal M_\theta(\mathcal J)=\{\pp_\rho : \rho\in \mathcal D(\mathcal H), \pp_\rho\circ\theta^{-1}=\pp_\rho\}.$$
The map $\rho\mapsto \pp_\rho$ is an affine. It has has a specific convex structure.
\begin{proposition}\label{prop:simplex}
    The set $\mathcal M_{\theta}(\mathcal J)$ is a convex simplex.
\end{proposition}
This is proved in \Cref{sec:proof sectors}. Given this structure we can define the set of sectors.
\begin{definition}[Sectors]
    The set $\mathcal S$, of sectors, is the set of extreme points of $\mathcal M_\theta(\mathcal J)$.
\end{definition}
\Cref{prop:simplex} ensures $\mathcal S$ is a finite set. To lighten the notation, we identify $\mathcal S$ with a subset of $\nn$. Then, for any element $\alpha\in \mathcal S$, we denote $\pp_\alpha$ the corresponding element in $\mathcal M(\mathcal J)$ and reciprocally. 

These measures verify some ergodic property.
\begin{theorem}
    \label{thm:sectors ergodic}
    For each $\alpha\in \mathcal S$, $\pp_\alpha$ is $\theta$-ergodic, therefore all the measures $\pp_\alpha$, $\alpha\in \mathcal S$ are two by two mutually singular.
\end{theorem} 
This theorem is proved in \Cref{sec:proof sectors}. Using it, we can construct a useful related partition of $\Omega$.
\begin{theorem}
    \label{thm:sectors and partition}
    There exist mutually disjoint $\theta$-invariant sets $\{\Omega_\alpha\}_{\alpha\in \mathcal S}$ such that
    \begin{enumerate}
        \item $\pp_\alpha(\Omega_\beta)=\delta_{\alpha,\beta}$, for any $\alpha,\beta\in \mathcal S$,
        \item for any $\rho\in \mathcal D(\mathcal H)$, $\pp_\rho(\cup_{\alpha\in \mathcal S}\Omega_\alpha)=1$.
    \end{enumerate}
\end{theorem}
This theorem is also proved in \Cref{sec:proof sectors}. It allows for the introduction of a random variable tracking the partition,
$$\Gamma=\sum_{\alpha\in \mathcal S}\alpha \one_{\Omega_\alpha}.$$

\subsection{Related POVM and subspaces}
We constructed the sectors from the probability measures induced by the instrument $\mathcal J$. In that sense they are functions of $\mathcal J$. We now relate the sectors to subspaces and a POVM on $\mathcal H$.

Preliminarily, let $\mathcal E_\Phi$ be the set of $\Phi$-invariant effects. For any $\rho\in \mathcal D(\mathcal H)$ and $E\in \mathcal E_\Phi$ such that $\tr(\rho E)>0$, and $\mathbf{a}\in \Omega_{\rm fin.}$, let
$$\pp_{\rho,E}(C_{\mathbf{a}})=\pp_{\rho,E}(\mathbf{a})=\tfrac{1}{\tr(\rho E)}\tr(\rho\Phi_{\mathbf{a}}(E)).$$
Since $E$ is $\Phi$-invariant, Kolmogorov's extension theorem ensures $\pp_{\rho,E}$ is a well defined probability measure on $\Omega$. The interpretation of $\pp_{\rho,E}$ is that, after some fixed number of measurements using the instrument $\mathcal J$, the POVM $\{E,\id_{\mathcal H}-E\}$ is performed and $\pp_{\rho,E}$ is the law of the outcomes of the $\mathcal J$ measurements conditioned on the POVM outcome leading to $E$. This can be generalized to any POVM whose positive operators are all $\Phi$-invariant. The measure $\pp_{\rho,E}$ is fictitious, since it is defined irrespectively of the number of $\mathcal J$ measurements. It will however reveal useful in understanding what are the sectors we just defined at measurement outcome level.
Using this construction we will identify a POVM and subspaces related to sectors.

Let
$$\mathcal T=\{ x\in \mathcal H : \langle x,{\Phi^*}^n(\rho)x\rangle\xrightarrow[n\to\infty]{}0,\forall \rho\in \mathcal D(\mathcal H)\}.$$
Cauchy-Schwartz inequality implies it is a subspace of $\mathcal H$. It is called the transient subspace. 

For any subspace $\mathcal K$, let $\mathcal K^\perp$ denote its orthogonal complement in $\mathcal H$.
\begin{theorem}
    \label{thm:sectors POVM}
    There exists a unique POVM $\{E_\alpha\}_{\alpha\in \mathcal S}$ such that
    \begin{enumerate}
        \item for each $\alpha\in \mathcal S$, $E_\alpha$ is $\Phi$-invariant,
        \item for any $\alpha\in \mathcal S$ and any $\rho\in \mathcal D_{\Phi^*}$ such that $\tr(\rho E_\alpha)>0$, $\pp_{\rho,E_\alpha}=\pp_\alpha$.
    \end{enumerate}

    Moreover, setting $\mathcal K_\alpha=E_\alpha\mathcal T^\perp$, for any $\alpha\in \mathcal S$, $\Phi^*(\mathcal D(\mathcal K_\alpha))\subset \mathcal D(\mathcal K_\alpha)$, and $\{\mathcal K_\alpha\}_{\alpha\in \mathcal S}$ is an orthogonal partition of $\mathcal T^\perp$.
\end{theorem}
This theorem is proved in \Cref{sec:proof sectors}. 

\begin{remark}
    The proof shows the operators $E_\alpha$ are the absorption operators related to the invariant subspaces (or enclosures) $\mathcal K_\alpha$ as defined in \cite{carbone2021absorption}.
\end{remark}

We have now all the elements to formulate our main results. In the rest of the article we use the shorthand $\pp_{\rho,\alpha}$ for $\pp_{\rho,E_\alpha}$.

\section{Main results}\label{sec:main results}
Our main results concerns the behavior in large time for the system state to be in one of the sector subspaces. The POVM $\{E_\alpha\}_{\alpha\in \mathcal S}$ introduced in \Cref{thm:sectors POVM} can be interpreted as a measurement of which sector. Then, for each $\alpha\in \mathcal S$ we define
$$Q_n(\alpha)=\tr(\rho_n E_\alpha),$$
where $\rho_n$ is defined in \Cref{eq:defrhon}, as the probability to be in sector $\alpha$ after $n$ measurements using instrument $\mathcal J$. Using Baye's rule, this quantity can also be interpreted as the probability of obtaining outcome $\alpha$ in a measurement of POVM $\{E_\beta\}_{\beta\in \mathcal S}$, conditioned on the first $n$ outcomes of the measurements using $\mathcal J$. Indeed, direct algebraic computations lead to
$$Q_n(\alpha)=\frac{\pp_{\rho_0,\alpha}(\omega_1,\dotsc,\omega_n)\tr(\rho_0E_\alpha)}{\pp_{\rho_0}(\omega_1,\dotsc,\omega_n)}.$$
A related quantity is one where the initial state is unknown and therefore replaced by a trial state $\hat \rho$. The requirement on this trial state is that it is positive definite so that every measurement outcome of any POVM on it is strictly positive. That ensures for example that $\pp_\rho$ is absolutely continuous with respect to $\pp_{\hat \rho}$. Then, similarly to \Cref{eq:defrhon}, one can defined the updated trial state given the first $n$ outcomes of the measurements using $\mathcal J$,
$$\hat \rho_n=\frac{\Phi^*_{a_n}\circ\dotsb \circ\Phi^*_{a_1}(\hat \rho)}{\pp_{\hat \rho}(a_1,\dotsc,a_n)}.$$
Then, the process $(\hat \rho_n)_n$ is called a filter of $(\rho_n)_n$. Similarly to $Q_n(\alpha)$, for any $\alpha\in \mathcal S$ we define
$$\widehat{Q}_n(\alpha)=\tr(\hat \rho_n E_\alpha)=\frac{\pp_{\hat \rho,\alpha}(\omega_1,\dotsc,\omega_n)\tr(\hat \rho_0E_\alpha)}{\pp_{\hat \rho}(\omega_1,\dotsc,\omega_n)}.$$
Our results show that $(Q_n(\alpha))_n$ and $(\widehat{Q}_n(\alpha))_n$ have similar behavior when $n$ grows.

Our first theorem demonstrates an almost sure exponentially fast selection of a random sector equivalent to an initial measurement of the POVM $\{E_\alpha\}_{\alpha\in \mathcal S}$ given by \Cref{thm:sectors POVM}. Moreover, the induced sector measurement result does not depend on the knowledge of the initial state since the filter converges towards the same sector at the same rate. Next theorem and corollary are both proved in \Cref{sec:as}.
\begin{theorem}\label{thm:generalized QND}
    For any $\rho\in \mathcal D(\mathcal H)$, the limits
    $$Q_\infty(\alpha)=\lim_{n\to\infty} Q_n(\alpha)\quad \mbox{and}\quad \widehat{Q}_\infty(\alpha)=\lim_{n\to\infty}\widehat{Q}_n(\alpha)$$
    exist $\pp_\rho$ almost surely and
    $$Q_\infty(\alpha)=\widehat{Q}_\infty(\alpha)=\one_{\Omega_\alpha}, \quad \pp_\rho\as$$
    with $\pp_{\rho}(Q_\infty(\alpha)=1)=\pp_\rho(\Gamma=\alpha)=\pp_\rho(\Omega_\alpha)=Q_0(\alpha)=\tr(E_\alpha\rho)$. 

    Moreover, for all $\alpha,\gamma\in\mathcal S$, setting $s(\gamma|\alpha)$ as the specific relative entropy
    $$s(\gamma|\alpha)=\lim_{n\to\infty}-\tfrac1n\mathbb E_\gamma\left(\ln\frac{\mathbb P_{\rho^{\rm ch},\alpha}(\omega_1,\dotsc,\omega_n)}{\mathbb P_\gamma(\omega_1,\dotsc,\omega_n)}\right),$$
    it cancels if and only if $\alpha=\gamma$ and,
    $$\limsup_{n\to\infty} \frac 1n\ln Q_n(\alpha)\leq -s(\Gamma|\alpha),\quad \pp_\rho\as$$
    and
    $$\limsup_{n\to\infty} \frac 1n\ln \widehat{Q}_n(\alpha)\leq -s(\Gamma|\alpha),\quad \pp_\rho\as$$
\end{theorem}
The following corollary straightens the interpretation that the system state converges towards one of the sectors with law given by the POVM $\{E_\alpha\}_{\alpha_\in \mathcal S}$ introduced in \Cref{thm:sectors POVM}. It shows that asymptotically, the system state is supported by one of the subspaces $\{\mathcal K_\alpha\}_{\alpha\in \mathcal S}$.
\begin{corollary}\label{cor:exp to subspace}
         $$\lim_{n\to\infty}\tr(P_{\Gamma}\rho_n)=\lim_{n\to\infty}\tr(P_{\Gamma}\hat\rho_n)=1, \quad \pp_\rho\as$$
    with $P_\Gamma$ the orthogonal projector onto $\mathcal K_\Gamma$ and $\pp_\rho(\Gamma=\alpha)=\tr(\rho E_\alpha)$ for all $\alpha\in \mathcal S$. Moreover, for any $\alpha\in \mathcal S$,
$$\limsup_{n\to\infty} \frac 1n\ln \tr(P_\alpha\rho_n)\leq -s(\Gamma|\alpha),\quad \pp_\rho\as$$
    and
    $$\limsup_{n\to\infty} \frac 1n\ln \tr(P_\alpha\widehat\rho_n)\leq -s(\Gamma|\alpha),\quad \pp_\rho\as$$
\end{corollary}

Next theorem expresses that the sector selection is also exponentially fast in mean. 
\begin{theorem}\label{thm:meanconv} There exists $0\leq\kappa<1$, $\tau\geq 1$ such that for all $\rho\in \mathcal D(\mathcal H)$ and $n\in\mathbb N$
    $$\mathbb E_\rho\left[\sum_{\alpha\neq\beta}\sqrt{Q_n(\alpha)Q_n(\beta)}\right]\leq \tau\sum_{\alpha\neq\beta}\sqrt{Q_0(\alpha)Q_0(\beta)}\ \kappa^n$$
    and for any positive definite $\hat \rho\in \mathcal D(\mathcal H)$,
    $$\mathbb E_\rho\left[\sum_{\alpha\neq\beta}\sqrt{\widehat{Q}_n(\alpha)\widehat{Q}_n(\beta)}\right]\leq \tau\|\hat \rho^{-\frac12}\rho \hat \rho^{-\frac12}\|_\infty \sum_{\alpha\neq\beta}\sqrt{\widehat{Q}_0(\alpha)\widehat{Q}_0(\beta)}\ \kappa^n.$$
\end{theorem}
Note that, the quantity $\sum_{\alpha\neq\beta}\sqrt{Q_n(\alpha)Q_n(\beta)}$ which serves as Lyapounov function is related to Réyni's relative entropy and Hellinger's distance. It is often used in the context of Bayesian Theory -- see \emph{e.g.} \cite{vaart2000asymptotic}.

We have again a corollary expressing the selection of $\mathcal K_\alpha$.

\begin{corollary}\label{coroo} There exists $0\leq\kappa<1$, $\tau\geq 1$ such that for all $\rho\in \mathcal D(\mathcal H)$ and $n\in\mathbb N$
    $$\mathbb E_\rho\left[\sum_{\alpha\neq\beta}\sqrt{\tr(P_\alpha\rho_n) \tr(P_\beta\rho_n)}\right]\leq \tau\sum_{\alpha\neq\beta}\sqrt{Q_0(\alpha)Q_0(\beta)}\ \kappa^n$$
    and for any positive definite $\hat \rho\in \mathcal D(\mathcal H)$,
    $$\mathbb E_\rho\left[\sum_{\alpha\neq\beta}\sqrt{\tr(P_\alpha\hat\rho_n)\tr(P_\beta\hat\rho_n)}\right]\leq \tau\|\hat \rho^{-\frac12}\rho \hat \rho^{-\frac12}\|_\infty \sum_{\alpha\neq\beta}\sqrt{\widehat{Q}_0(\alpha)\widehat{Q}_0(\beta)}\ \kappa^n.$$
\end{corollary}

\section{Invariant states and POVMs, \Cref{thm:sectors POVM,thm:sectors ergodic,thm:sectors and partition} proofs}\label{sec:proof sectors}
As a preliminary to these proofs,  we establish some technical results.
\begin{lemma}
    \label{lem:pp lipshitz}
    For any $\rho,\sigma\in \mathcal D(\mathcal H)$,
    $$\sup_{A\in \mathcal F}|\pp_\rho(A)-\pp_\sigma(A)|\leq \|\rho-\sigma\|_{\rm tr}$$
    with $\|\cdot\|_{\rm tr}$ the trace norm.
\end{lemma}
\begin{proof}
    For $n\in \nn$, let $P_n: \mathcal F_n\to \mathcal B(\mathcal H)$ be defined by
    $$P_n(A)=\sum_{\mathbf{a}\in \mathrm{ind}_A}\Phi_{a_1}\circ\dotsb\circ\Phi_{a_n}(\id)$$
    where $\mathrm{ind}_A\subset \mathcal A^n$ is such that $A=\cup_{\mathbf a\in \mathrm{ind}_A}{\mathbf{a}}$. Since $\Phi(\id)=\sum_{a\in \mathcal A}\Phi_a(\id)=\id$,
    $$P_{n+1}(A)=P_n(A)$$
    for any $A\in \mathcal F_n$. Thus, by Kolmogorov extension theorem for POVMs \cite[Corollary~1]{tumulka2008kolmogorov}, there exists a POVM $P:\mathcal F\to \mathcal B(\mathcal H)$, such that
    $$\pp_{\rho}(A)=\tr(\rho P(A))$$
    for any $A\in \mathcal F$.
    
    Then, since $P(A)\leq \id$, using Holder's inequality for matrix Schatten norms, $|\tr(X^*Y)|\leq \|X\|_{\rm tr}\|Y\|_\infty$ for any matrix $X, Y$ and
    $$|\pp_{\rho}(A)-\pp_\sigma(A)|\leq \|\rho-\sigma\|_{\rm tr}\ \|P(A)\|_\infty\leq \|\rho-\sigma\|_{\rm tr}$$
    for any $A\in \mathcal F$. That concludes the theorem proof.
\end{proof}
With respect to the left shift on $\Omega$:
$$\theta(\omega_1,\omega_2,\dotsc)=(\omega_2,\omega_3,\dotsc),$$
the measures $\pp_\rho$ have strong convergence properties.
\begin{lemma}
    \label{lem:pp Cesaro convergence}
    Let $$T_\infty:=\lim_{n\to\infty} \frac1n\sum_{k=1}^{n}{\Phi^*}^k.$$ Then,
    $$\lim_{n\to\infty}\sup_{\rho\in \mathcal D(\mathcal H)}\sup_{A\in \mathcal F}\left|\frac1n\sum_{k=1}^n\pp_\rho\circ\theta^{-k}(A) - \pp_{T_\infty(\rho)}(A)\right|=0.$$
\end{lemma}
\begin{proof}
    First, since $\Phi^*$ is a positive map preserving the trace, $T_\infty$ is well defined and the convergence to it holds in norm -- see \cite[Proposition~6.3]{Wolf}.
    
    Second, for $A=C_{\mathbf{a}}$ for some $\mathbf{a}\in \Omega_{\rm fin.}$, 
    $$\pp_\rho\circ\theta^{-1}({\mathbf{a}})=\sum_{b\in \mathcal A}\pp_\rho(b,a_1,\dotsc,a_p).$$
    Then,
    $$\pp_\rho\circ\theta^{-1}({\mathbf{a}})=\tr(\rho\Phi\circ\Phi_{a_1}\circ\dotsb\circ\Phi_{a_p}(\id)).$$
    Hence, $\pp_\rho\circ\theta^{-1}=\pp_{\Phi^*(\rho)}$. Then, since $\rho\mapsto \pp_\rho$ is affine by definition,
    $$\frac1n\sum_{k=1}^n\pp_\rho\circ\theta^{-k}=\pp_{\frac1n\sum_{k=1}^n{\Phi^*}^k(\rho)}.$$
    Finally, \Cref{lem:pp lipshitz} implies that
    $$\sup_{A\in \mathcal F}\left|\frac1n\sum_{k=1}^n\pp_\rho\circ\theta^{-k}(A) - \pp_{T_\infty(\rho)}(A)\right|\leq\left\Vert\frac1n\sum_{k=1}^n{\Phi^*}^k(\rho)-T_{\infty}(\rho)\right\Vert_{\textrm{tr}}$$
    and the fact that the convergence to $T_\infty(\rho)$ is uniform in $\rho$ yield the lemma.
\end{proof}

Let us now introduce a suitable decomposition of $\mathcal J$ and $\mathcal H$ with respect to the fixed points of $\Phi^*$.
Following \cite[Theorem~6.14]{Wolf}, the set $\mathcal F_{\Phi^*}$ of fixed points of $\Phi^*$ is given by
\begin{equation}\label{eq:def fixed points}
    \mathcal F_{\Phi^*}=U\left(0_{d_0}\oplus\bigoplus_{k=1}^K M_{d_k}(\cc)\otimes \varrho_k\right)U^*
\end{equation}
with $U$ a unitary operator on $\mathcal H$, $\{\varrho_k\}_{k=1}^K$ a set of positive definite density matrices of respective dimensions $m_k\times m_k$ such that $d_0+d_1+\dotsb+d_K+m_1+\dotsb+m_K=\dim \mathcal H$ and $\dim \mathcal T=d_0$.

For $k\in\{1,\dotsc,K\}$ and $u_k\in S^{d_k-1}(\cc)$, let 
$$\mathcal H_{k,u}=\operatorname{range} U\left(0_{d_0}\oplus \left(\bigoplus_{l<k} 0_{d_l}\otimes 0_{m_l}\right)\oplus(u_ku_k^*\otimes \varrho_k)\oplus \left(\bigoplus_{l>k} 0_{d_l}\otimes 0_{m_l}\right)\right)U^*.$$
Then \Cref{eq:def fixed points} and the positivity of $\Phi^*$ imply $\Phi^*(\mathcal B(\mathcal H_{k,u}))\subset\mathcal B(\mathcal H_{k,u})$. 
Again, positivity implies that for any $a\in \mathcal A$,
$$\Phi_a^*(\mathcal B(\mathcal H_{k,u}))\subset \mathcal B(\mathcal H_{k,u}).$$

For simplicity, from now on, we work in a basis such that $U=\id_{\mathcal H}$ and summarize all the $0$ matrices in one notation $0$ when the context is clear. Then, by linearity, for any $k\in\{1,\dots,K\}$, $a\in \mathcal A$ and $x_k\in M_{d_k}(\cc)$,
$$\Phi_a^*((x_k\otimes M_{m_k}(\cc))\oplus 0)\subset(x_k\otimes M_{m_k}(\cc))\oplus 0.$$
Thus, for any $a\in \mathcal A$, there exists $\Psi_a:\mathcal B(\mathcal H)\to \mathcal B(\mathcal T)$ and for any $k\in\{1,\dotsc,K\}$, $\Phi_{a,k}: M_{m_k}(\cc)\to M_{m_k}(\cc)$ such that
\begin{align}\label{eq:decomp Phi_a}
    \Phi_a\equiv\Psi_a\oplus\bigoplus_{k=1}^K \left(\id_{M_{d_k}(\cc)}\otimes \Phi_{a,k}\right).
\end{align}
Moreover, each $\Phi_k=\sum_{a\in \mathcal A}\Phi_{a,k}$ is an irreducible completely positive unital map from $M_{m_k}(\cc)$ to itself since $\varrho_k$ is positive definite and is the unique $\Phi_k^*$-invariant state.

We now turn to the proof of existence of a POVM verifying the two conditions of \Cref{thm:sectors POVM}. The decompositions of $\mathcal J$ and $\mathcal H$ we just introduced will lead us to the construction of the POVM and related subspaces. We prove \Cref{prop:simplex,thm:sectors ergodic,thm:sectors and partition} along the way.

Let $\{u_{k,i}\}_{i=1}^{d_k}$be an orthonormal basis of $\cc^{d_k}$ and $\mathcal H_k=\bigoplus_{i=1}^{d_k}\mathcal H_{k,u}$. Then, for any two $\Phi^*$-invariant state $\varrho, \varrho'$ such that $\supp\varrho\subset \mathcal H_k$ and $\supp\varrho'\subset \mathcal H_k$, \Cref{eq:decomp Phi_a} implies
$$\pp_\varrho=\pp_{\varrho'}.$$
Let us denote this common shift invariant measure $\pp_k$. 

Assume $\pp_\rho$ is $\theta$-invariant (\emph{i.e.} $\pp_\rho\in \mathcal M_\theta(\mathcal J)$). Then, \Cref{lem:pp Cesaro convergence} implies $\pp_\rho=\pp_{T_\infty(\rho)}$. Hence, there exist $\varrho\in \mathcal D_{\Phi^*}$ such that $\pp_\rho=\pp_{\varrho}$. Since $\rho\mapsto \pp_\rho$ is affine, \Cref{eq:def fixed points} implies $\mathcal M_\theta(\mathcal J)$ is the convex hull of $\{\pp_k\}_{k=1}^K$ and \Cref{prop:simplex} follows.

The measure $\pp_k$ relates the statistics of $\mathcal J_k=\{\Phi_{a,k}\}_{a\in \mathcal A}$ with respect to $\varrho_k$. Since $\Phi_k$ is irreducible \cite[Corollary~5]{Maassen2000AnET} implies $\pp_k$ is $\theta$-ergodic and \Cref{thm:sectors ergodic} is proved.

Since they are $\theta$-ergodic, for each $k\in\{1,\dotsc,K\}$, $\pp_k$ is an extreme point of $\mathcal M_\theta(\mathcal J)$.
Let $k\sim k'$ if and only if $\pp_k=\pp_{k'}$. Then, the set of sectors $\mathcal S$ is in bijection with $\{1,\dotsc,K\}/\sim$.

Item (1) of \Cref{thm:sectors and partition} is a direct consequence of \Cref{thm:sectors ergodic}. 
For Item~(2), fix $\rho\in \mathcal D(\mathcal H)$. \Cref{lem:pp Cesaro convergence} and the $\theta$-invariance of $\cup_{\alpha\in \mathcal S}\Omega_\alpha$ imply
$$\pp_\rho(\cup_{\alpha\in\mathcal S}\Omega_\alpha)=\pp_{T_\infty(\rho)}(\cup_{\alpha\in\mathcal S}\Omega_\alpha).$$
Then, $\pp_{T_\infty(\rho)}\in \mathcal M_\theta(\mathcal J)$ yields \Cref{thm:sectors and partition}.

For \Cref{thm:sectors POVM}, for any $\alpha\in \mathcal S$, let $\mathcal K_\alpha=\bigoplus_{k: \pp_k=\pp_\alpha} \mathcal H_k$.
Then, by definition, $\mathcal H=\mathcal T\oplus \bigoplus_{\alpha\in \mathcal S}\mathcal K_\alpha$.
It follows that each $\mathcal K_\alpha$ is a $\Phi$-invariant subspace or an enclosure in the language of \cite{Baumgartner2012structures, carbone2016irreducible,carbone2021absorption}.

Then, following \cite[Proposition~6]{carbone2021absorption}, 
$$E_\alpha=\lim_{n\to\infty}\Phi^n(P_\alpha),$$
where $P_\alpha$ is the orthogonal projector onto $\mathcal K_\alpha$, is an absorption operator. Therefore $E_\alpha$ is $\Phi$-invariant and Item~(1) of \Cref{thm:sectors POVM} holds.

Fix $\rho\in \mathcal D(\mathcal H)$. By definition of $\mathcal T$,
$$1=1-\lim_{n\to\infty}\tr(\rho\Phi^n(P_{\mathcal T}))=\lim_{n\to\infty}\sum_{\alpha\in \mathcal S}\tr(\rho\Phi^n(P_\alpha))=\sum_{\alpha\in \mathcal S}\tr(\rho E_\alpha)$$
with $P_{\mathcal T}$ the orthogonal projector onto $\mathcal T$.
It follows, $\sum_{\alpha\in \mathcal S}E_\alpha=\id_{\mathcal H}$. Hence $\{E_\alpha\}_{\alpha}$ is a POVM.

Assume $\rho\in \mathcal D(\mathcal H)$ is $\Phi^*$-invariant. By definition of the subspaces $\mathcal K_\alpha$ and \Cref{eq:def fixed points},
it is a convex combination of invariant states $\{\rho_\alpha\}_{\alpha\in \mathcal S}$ with ranges included in $\mathcal K_\alpha$ respectively.
Since \cite[Proposition~6]{carbone2021absorption} implies $P_\alpha E_\beta P_\alpha=\delta_{\alpha,\beta}P_\alpha$, using $\Phi_a^*(\mathcal B(\mathcal K_\alpha))\subset\mathcal B(\mathcal K_\alpha)$,
$$\pp_{\rho,E_\alpha}=\pp_{\rho_\alpha}.$$
Since, for any invariant state $\rho_\alpha$ with range included in $\mathcal K_\alpha$, $\pp_{\rho_\alpha}=\pp_\alpha$,
Item~(2) of \Cref{thm:sectors POVM} holds.

Concerning the subspaces $\mathcal K_\alpha$, using again \cite[Proposition~6]{carbone2021absorption}, $E_\alpha=P_\alpha+T_\alpha$ where $T_\alpha$ is a positive semi-definite operator whose range is orthogonal to $\mathcal K_\alpha$.
Since $\mathcal K_\alpha\perp\mathcal K_\beta$ for any $\alpha\neq \beta$ by construction, the same proposition yields that actually the range of $T_\alpha$ is included in $\mathcal T$. And by construction, $\mathcal K_\alpha$ is an invariant subspace (or enclosure) and $\{\mathcal K_\alpha\}_{\alpha\in \mathcal S}$ is an orthogonal partition of $\mathcal T^\perp$.

We now turn to the proof of the uniqueness of the POVM. Assume $\{N_\alpha\}_{\alpha\in \mathcal S}$ is a POVM verifying (1-2). By Item~(1), $\Phi(N_\alpha)=N_\alpha$, which implies $\pp_{\rho,N_\alpha}$ is well defined when $\tr(N_\alpha\rho)>0$. Item~(2) implies then, $\pp_{\rho,N_\alpha}=\pp_\alpha$ for any $\rho\in \mathcal D_{\Phi^*}$ such that $\tr(\rho N_\alpha)>0$. Since by definition $\pp_\rho=\sum_{\alpha\in \mathcal S}\tr(N_\alpha\rho)\pp_{\rho, N_\alpha}$, it follows that for any $\rho\in \mathcal D_{\Phi^*}$,
$$\pp_\rho=\sum_{\alpha\in \mathcal S}\tr(\rho N_\alpha) \pp_\alpha.$$
Therefore, $\pp_\alpha\perp\pp_\beta$ for any two distinct $\alpha,\beta\in \mathcal S$ implies $\tr(N_\alpha\rho)=\tr(E_\alpha\rho)$ for any $\rho\in \mathcal D_{\Phi^*}$. Let $\rho\in \mathcal D(\mathcal H)$ be arbitrary, by $\Phi$-invariance of $N_\alpha$ and $E_\alpha$, $T_\infty^*(N_\alpha)=N_\alpha$ and $T_\infty^*(E_\alpha)=E_\alpha$, with $T_\infty$ defined in \Cref{lem:pp Cesaro convergence}. Since $T_\infty(\mathcal D(\mathcal H))=\mathcal D_{\Phi^*}$,
$$\tr(N_\alpha\rho)=\tr(N_\alpha T_\infty(\rho))=\tr(E_\alpha T_\infty(\rho))=\tr(E_\alpha\rho).$$
Hence, for any $\alpha\in\mathcal S$, $N_\alpha=E_\alpha$ and
 \Cref{thm:sectors POVM} is proved. \hfill\qed

\begin{remark}
    Birkhoff's ergodic theorem implies the sets $\Omega_\alpha$ can be chosen as
    $$\Omega_\alpha=\left\{\omega:\lim_n\tfrac1nN_n(\mathbf{a})=\pp_\alpha(\mathbf{a}), \forall \mathbf{a}\in \Omega_{\rm fin.}\right\}$$
    with $N_n(\mathbf{a})=\operatorname{Card}\{1\leq k\leq n-|a|+1: \omega_k=a_1,\dotsc,\omega_{k+|a|-1}=a_p\}$ where $|\mathbf{a}|$ is the length of $\mathbf{a}$.
\end{remark}

\section{Examples}\label{sec:xpl}
Before we delve into the proofs concerning the convergence, we illustrate our results with a few examples.

\subsection{Irreducible channels}
Assume $\Phi$ is irreducible. Then, by Perron-Frobenius Theorem -- see \cite{EHK} -- there exist a unique $\Phi^*$-invariant state $\varrho$ and it is positive definite. Hence both $\mathcal D_{\Phi^*}$ and $\mathcal E_\Phi$ are singleton and $\mathcal K=\mathcal H$ is the unique sector. Hence, the convergence is instantaneous.

\subsection{Identity channel}\label{sec:xpl identity} 
We present a drastically different example where there is a unique sector. Consider the identity channel 
$$\Phi:X\mapsto X.$$
Any associated instrument is given by $\Phi_a=p_a\Phi$ for $(p_a)_{a\in \mathcal A}$ a probability vector.

The set of fixed points of $\Phi^*$ is the whole algebra $\mathcal B(\mathcal H)$.
Since $\mathbb P_{\varrho}=\mathbb P_{\varrho'}$ for any $\varrho,\varrho'\in \mathcal D(\mathcal H)$, there is only one sector $\mathcal K=\mathcal H$ and the convergence is instantaneous.

\subsection{Quantum non-demolition measurement and generalization}\label{sec:no transient}\label{sec:xpl qnd}
In this example and the following ones, we focus on perfect instruments. Let us recall their definition. Stinespring's theorem implies there exist a finite alphabet $\mathcal A$ and $(K_a)_{a\in \mathcal A}\in \mathcal B(\mathcal H)^{\mathcal A}$ such that
$$\Phi(X)=\sum_{a\in \mathcal A} K_a^*X K_a$$
with $\sum_{a\in \mathcal A} K_a^*K_a=\textrm{Id}_{\mathcal H}$. Then
$$\Phi_a(X)=K_a^*XK_a,\quad a\in \mathcal A$$
defines an instrument $\mathcal J=\{\Phi_a\}_{a\in \mathcal A}$. Such an instrument is called perfect since eauch $\Phi_a$ has, at most, Kraus rank $1$. General instruments can always be obtained through sums of convex combinations of perfect instruments.

The present example is the one studied in \cite{amini2023exponential} and is a generalization of the QND model studied in \cite{BenrardBauer,BenoistAHP}. In the QND model all the subspaces $\mathcal H_\alpha$ we shall introduce are all one dimensional.

Consider block diagonal Kraus operators. Namely, $\mathcal H=\cc^d$ and
$$
K_a=\left(\begin{array}{cccc}
  K_{1,a}   & 0&\cdots &0 \\
   0  & \ddots&\ddots&\vdots\\
   \vdots&\ddots&\ddots&0\\
   0&\cdots&0&K_{m,a}
\end{array}\right)
$$
where $K_{i,a} \in M_{d_i}(\mathbb{C})$, for $i=1,\ldots,m$. This corresponds to a decomposition of $\mathcal H=\cc^d$ into orthogonal subspaces
$$
\mathcal{H}=\bigoplus_{i=1}^{m} \mathcal{H}_i, \quad \text{where } \mathcal{H}_i \equiv \mathbb{C}^{d_i}.
$$
For $i=1,\dotsc,m$, let $\Phi^{(i)}:\mathcal B(\mathcal H_i)\to \mathcal B(\mathcal H_i)$ be defined by
$$\Phi^{(i)}:X\mapsto \sum_{a\in \mathcal A}K_{i,a}^*XK_{i,a}.$$
Since $\Phi(\id_{\mathcal H})=\id_{\mathcal H}$,
$$
\Phi_i(\id_{\mathcal H_i}) = \id_{\mathcal H_i},
$$
so, it is a quantum channel on $\mathcal B(\mathcal H_i)$. It is such that for any $X\in \mathcal B(\mathcal H_\alpha)$,
$$\Phi(X\oplus 0_{\mathcal B(\oplus_{j\neq i}\mathcal H_j)})=\Phi_\alpha(X)\oplus 0_{\mathcal B(\oplus_{j\neq i}\mathcal H_j)}.$$

Assume that for $i=1,\dotsc,m$, $\Phi_i$ is irreducible with unique positive definite invariant state $\varrho_i$. 
By abuse of notation we also denote $\varrho_i$ the state in $\mathcal D(\mathcal H)$ defined by $\varrho_i\oplus 0_{\mathcal B(\oplus_{j\neq i}\mathcal H_j)}$. 

We assume that the probability measures with respect to invariant states satisfy
$$
\mathbb{P}_{\varrho_i} \neq \mathbb{P}_{\varrho_j}, \quad \text{for } i \neq j.
$$
Then the sectors subspaces are given by
$$
\mathcal{K}_{\alpha} = \operatorname{supp} \varrho_\alpha=\mathcal H_\alpha, \quad \alpha \in \mathcal{S},
$$
where $\mathcal{S} = \{1, \ldots, m\}.$

The main assumption in the present example is the absence of transient part, that is, $\mathcal{T} = \{0\}$. Let  $\rho^{\rm ch} = \frac{\textrm{Id}_\mathcal{H}}{\textrm{dim} \mathcal{H}}$ be the initial state. Following the strategy of \Cref{thm:generalized QND} proof, Lemma~\ref{lem:L subadditive}, the fact there exist $c>0$ such that $c^{-1}\pp_\beta\leq \pp_{\rho^{\rm ch},\beta}\leq c \pp_\beta$ for any $\beta\in \mathcal S$ and Kingmann's sub-additive ergodic theorem leads to 
\begin{align*}
\lim_{n\to\infty}\tfrac1n\ln Q_n(\alpha)=&\lim_{n\to\infty} \tfrac1n\ln\frac{Q_n(\alpha)}{Q_n(\Gamma)}\\
=&\lim_{n\to\infty}\tfrac1n(L_n(\alpha)-L_n(\Gamma))\\
=&-\sup_n \frac{S(\pp_{\Gamma}|_{\mathcal F_n}|\pp_{\alpha}|_{\mathcal F_n})-C}{n}\\
=&-s(\Gamma|\alpha),
\end{align*}
$\pp_{\rho^{\rm ch}}$-almost surely for some $C>0$.

The rate $s(\gamma|\alpha)$ cannot be made more explicit in general as it comes from an application of Fekete's sub-additive lemma. Nevertheless, in the QND case, it can be made explicit since the measures $\pp_\beta$ are laws of i.i.d. random variables. It is one of the results of \cite{BenrardBauer,BenoistAHP}. In this case,  
$$
s(\gamma \vert \alpha) = \sum_{a\in \mathcal A} \vert K_{\gamma,a} \vert^2 \log \frac{\vert K_{\gamma,a} \vert^2}{\vert K_{\alpha,a} \vert^2},
$$  
which is the Kullback-Leibler divergence of $(|K_{\gamma,a}|^2)_{a\in \mathcal A}$ with respect to $(|K_{\alpha,a}|^2)_{a\in \mathcal A}$. 

Similarly, we can give explicit constant and rate in \Cref{thm:meanconv}. Indeed, one can chose $N=1$ in \Cref{Prop:identif}. Then \Cref{lem:kappa} implies
$$
\kappa = \max_{\alpha \neq \beta} \sum_{a\in \mathcal A} \sqrt{\vert K_{\alpha,a} \vert^2 \vert K_{\beta,a} \vert^2}.
$$
Following a remark at the end of \Cref{thm:meanconv} proof, one can set $\tau=1$ in that same theorem.

Note that the condition $\mathbb{P}_\alpha \neq \mathbb{P}_{\beta}$ is equivalent to the fact that for all $\alpha \neq \beta$, there exists $a\in \mathcal A$ such that $\vert K_{\alpha,a} \vert^2 \neq \vert K_{\beta,a} \vert^2$. This condition ensures that $s(\beta \vert \alpha) > 0$ for $\beta \neq \alpha$ and $0 \leq \kappa < 1$.

\subsection{Introducing a transient subspace.} 
Assume now that $\mathcal{T} \neq \{0\}$, so that $\mathcal H=\left(\oplus_{i=1}^m \mathcal H_i\right)\oplus \mathcal T$ and the Kraus operators are block matrices of the form  
$$  
K_a = \left( \begin{array}{cccc}  
  K_{1,a}   & 0 & \cdots & \star \\  
   0  & \ddots & \ddots & \vdots \\  
   \vdots & \ddots & K_{m,a} & \star \\  
   0 & \cdots & 0 & K_{\mathcal{T},a}  
\end{array} \right).  
$$  
The stars are non-zero matrices such that there does not exist a non zero invariant subspace contained in $\mathcal T$. 
The matrices $K_{\mathcal T,a}$ are of size $\dim \mathcal{T} \times \dim \mathcal{T}$. 
As in the previous example we assume that the quantum channels $\Phi_i$ defined on $\mathcal B(\mathcal H_i)$ by the Kraus operators $(K_{i,a})_{a\in \mathcal A}$ are all irreducible. We denote their invariant states $\varrho_i$. We considered them as elements of $\mathcal D(\mathcal H_i)$ or $\mathcal D(\mathcal H)$ supported on $\mathcal H_i$ indiscriminately. 

\medskip
If we assume $\pp_{\varrho_i}\neq\pp_{\varrho_j}$ for any $i\neq j$, the sectors are still given by  
$$  
\mathcal{K}_{\alpha} = \textrm{supp}\, \varrho_\alpha=\mathcal H_\alpha, \quad \alpha \in \mathcal{S},  
$$  
where $\mathcal{S} = \{1, \ldots, m\}$.

On the contrary, if there exist a phase $z\in U(1)$ and a unitary operator $U\in U(\mathcal H)$ such that $K_{2,a}=z UK_{1,a}U^*$, then $\pp_1=\pp_2$. Moreover, assume that for any $i\neq j$ such that either $i$ or $j$ is not in $\{1,2\}$, $\pp_i\neq\pp_j$. Then $\mathcal S=\{1,3,\dotsc, m\}$, the first sector is
$$\mathcal K_{1}=\mathcal H_1\oplus \mathcal H_2,$$
and all the other sectors are unchanged,
$$\mathcal K_{\alpha}=\mathcal H_{\alpha},\quad \forall \alpha>2.$$

\subsection{Quantum non-demolition with a transient space}\label{sec:xpl numerics}
In the last example we choose $m=2$ and $\dim\mathcal H_1=\dim\mathcal H_2=\dim\mathcal T=1$.

This example in dimension $3$ allows us to compare our rate of selection $s(\gamma|\alpha)$ from \Cref{thm:generalized QND} with known results and shows, through numerical simulations, that it can be smaller than the selection rate for non-demolition measurement and the escape rate from the transient subspace. 

Let $\mathcal H=\cc^3$, $\mathcal A=\{0,1\}$, $(p,q)\in \{(x,y)\in (0,1)^2 : 0<p+q<1, p\neq q\}$, $r=1-p-q$,
$$K_0=\begin{pmatrix}
    \sqrt{\frac{q}{p+q}}&0& \sqrt{p/2}\\
    0 & \sqrt{\frac{p}{p+q}}& \sqrt{q/2}\\
    0 & 0 & \sqrt{r/2}
\end{pmatrix},\quad K_1=\begin{pmatrix}
    -\sqrt{\frac{p}{p+q}}&0& \sqrt{q/2}\\
    0 & -\sqrt{\frac{q}{p+q}}& \sqrt{p/2}\\
    0 & 0 & \sqrt{r/2}
\end{pmatrix}.$$
The assumptions on $p$ and $q$ imply the sectors are given by $\mathcal S=\{1,2\}$ with invariant states 
$$\varrho_1=\begin{pmatrix}
    1&0&0\\
    0&0&0\\
    0&0&0
\end{pmatrix}\quad\mbox{and}\quad\varrho_2=\begin{pmatrix}
    0&0&0\\
    0&1&0\\
    0&0&0
\end{pmatrix}.$$
The associated effect operators are,
$$E_1=\begin{pmatrix}
    1&0&0\\
    0&0&0\\
    0&0&\frac{q}{p+q}
\end{pmatrix}\quad\mbox{and}\quad E_2=\begin{pmatrix}
    0&0&0\\
    0&1&0\\
    0&0&\frac{p}{p+q}
\end{pmatrix}.$$
Our goal is to provide estimates on the rate $s(1|2)$ from \Cref{thm:generalized QND} in that case. From the proof of \Cref{thm:generalized QND} it is given by 
$$s(1|2)=-\limsup_{n\to\infty}\tfrac{1}{n}\ln \frac{\pp_{\rho^{\rm ch},2}(\omega_1,\dotsc,\omega_n)}{\pp_1(\omega_1,\dotsc,\omega_n)}$$
with $(\omega_n)_{n\in \nn}$ distributed according to $\pp_1$. From the expression of $E_2$ and the definition of $\pp_{\rho^{\rm ch},2}$,
$$\pp_{\rho^{\rm ch},2}(\omega_1,\dotsc,\omega_n)=\tfrac{1}{1+\frac{p}{p+q}}\left(\pp_{2}(\omega_1,\dotsc,\omega_n)+\pp_{\varrho_3,2}(\omega_1,\dotsc,\omega_n)\right)$$
with $\varrho_3=\begin{pmatrix}
    0&0&0\\
    0&0&0\\
    0&0&1
\end{pmatrix}$. Again from the expression of $E_2$,
$$\pp_{\varrho_3,2}(\omega_1,\dotsc,\omega_n)=\frac{p}{p+q}(r/2)^n +|\langle e_2,K_{\omega_n}\dotsb K_{\omega_1}e_3\rangle|^2$$
with $\{e_1,e_2,e_3\}$ the canonical basis of $\cc^3$.

Let
$$h(1|2)=-\limsup_{n\to\infty}\tfrac1n\ln\frac{\pp_2(\omega_1,\dotsc,\omega_n)}{\pp_1(\omega_1,\dotsc,\omega_n)},$$
$$h(1)=-\limsup_{n\to\infty}\tfrac1n\ln\pp_1(\omega_1,\dotsc,\omega_n)$$
and
$$\tau(1|2,3)=-\limsup_{n\to\infty}\tfrac1n\ln\frac{|\langle e_2,K_{\omega_n}\dotsb K_{\omega_1}e_3\rangle|^2}{\pp_1(\omega_1,\dotsc,\omega_n)}.$$
Then,
$$s(1|2)=\min( h(1|2), -\log(r/2)-h(1),\tau(1|2,3)).$$
Since $\pp_1$ and $\pp_2$ are laws of sequences of i.i.d. random variables, the first two limit superior are limits and
$$h(1|2)=\tfrac{p}{p+q}\ln(p/q)+\tfrac{q}{p+q}\ln(q/p)\quad \mbox{and}\quad h(1)=-\tfrac{p}{p+q}\ln(\tfrac{p}{p+q})-\tfrac{q}{p+q}\ln(\tfrac{q}{p+q})$$
almost surely with respect to $\pp_1$.

Remark that $h(1|2)$ is a Kullback-Leibler divergence whereas $h(1)$ is a Shannon entropy.
Depending on the value of $p$ and $q$, both alternatives $h(1|2)\leq -\log(r/2)-h(1)$ and $h(1|2)\geq -\log(r/2)-h(1)$ are possible. One question is wether $s(1|2)=\tau(1|2,3)<\min(h(1|2),-\log(r/2)-h(1))$ is possible. Since $\tau(1|2,3)$ is not easily computable, we provide some numerical results in \Cref{fig:numerics}. There, one can see that this eventuality is possible, which hints that $s(\gamma|\alpha)$ is a priory hard to compute in full generality. However it is relatively easy to estimate numerically by simulating the appropriate process.

\begin{figure}[h!]
    \includegraphics[width=.45\linewidth]{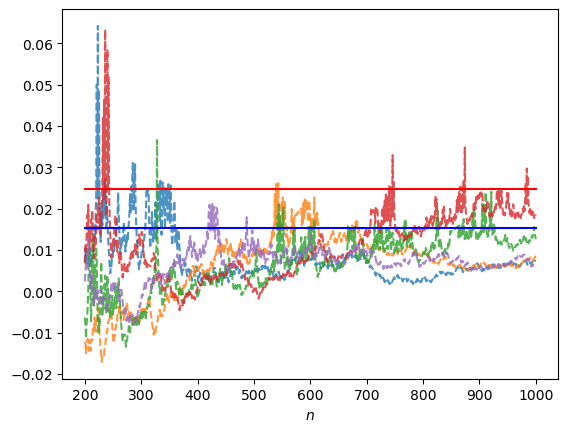}
    \includegraphics[width=.45\linewidth]{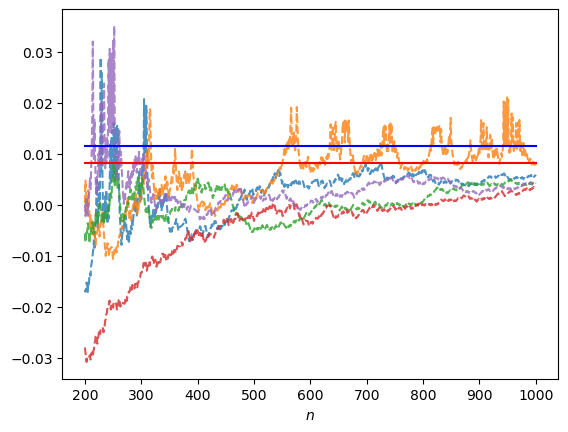}
    \caption{\label{fig:numerics} Numerical simulations for the example of \Cref{sec:xpl numerics}. Solid blue lines represent the value of $-\ln(r/2)-h(1)$. Solid red lines represent the value of $h(1|2)$. Dashed lines represent numerical computations of $\tfrac{1}{n}\ln\frac{|\langle e_2,K_{\omega_n}\dotsb K_{\omega_1}e_3\rangle|^2}{\pp_1(\omega_1,\dotsc,\omega_n)}$ for $200\leq n \leq 1000$. In the left panel $p=0.005$ and $q=0.004$. In the right panel $p=0.005$ and $q=0.0044$. In both cases we remark that $\tau(1|2,3)$ seems to be the smallest rate of convergence to $0$.}
\end{figure}

\section{Deformed instruments}\label{sec:deformed instruments}
Before we prove the main results, let us introduce the main new objects we define. For any $\alpha\in\mathcal S$, $\mathcal J^{(\alpha)}$ is an instrument on $\mathcal H_\alpha=E_\alpha\mathcal H=\mathcal K_\alpha\oplus E_\alpha \mathcal T$. 
\begin{definition}
    For any $a\in \mathcal A$, let $\Phi_a^{(\alpha)}:\mathcal B(\mathcal H_\alpha)\to \mathcal B(\mathcal H_\alpha)$ be defined by\footnote{We omit the canonical injection $\mathcal B(\mathcal H_\alpha)\to \mathcal B(\mathcal H)$ using $\mathcal B(\mathcal H_\alpha)\equiv E_\alpha\mathcal B(\mathcal H)E_\alpha$.}
    $$\Phi_a^{(\alpha)}(X)=E_\alpha^{-\frac12}\Phi_a(E_\alpha^{\frac12}XE_\alpha^{\frac12})E_\alpha^{-\frac12}$$
    with $E_\alpha^{-1}$ being the Penrose pseudo inverse of $E_\alpha$.
\end{definition}

\begin{proposition}\label{prop:deformed instrument}
    The indexed set $\{\Phi^{(\alpha)}_a\}_{a\in \mathcal A}$ is an instrument on $\mathcal B(\mathcal H_\alpha)$ with associated quantum channel $\Phi^{(\alpha)}$.
    Moreover, the set of fixed points of ${\Phi^{(\alpha)}}^*$ is $E_\alpha\mathcal F_{\Phi^*} E_\alpha$.
\end{proposition}
\begin{proof}
    The map $\Phi^{(\alpha)}$ is CP by construction, thus to prove it is a quantum channel it remains to prove it preserves the identity on $\mathcal{H}_\alpha$:
    \begin{align*}
        \Phi^{(\alpha)}(\id_{\mathcal{H}_\alpha}) =\sum_{a\in \mathcal{A}}\Phi_a^{(\alpha)}(\id_{\mathcal{H}_\alpha})
        &= \sum_{a\in \mathcal{A}} E_\alpha^{-\frac12}\Phi_a(E_\alpha^{\frac12}\id_{\mathcal{H}_\alpha}E_\alpha^{\frac12})E_\alpha^{-\frac12} \\
        &= E_\alpha^{-\frac12} \sum_{a\in \mathcal{A}} \Phi_a(E_\alpha) E_\alpha^{-\frac12} \\
        &= E_\alpha^{-\frac12}  \Phi(E_\alpha) E_\alpha^{-\frac12} = \id_{\mathcal{H}_\alpha}, \\
    \end{align*}
    where we used that $E_\alpha$ is and invariant effect of the quantum channel $\Phi$. 
    
    It remains to derive the set of fixed points of ${\Phi^{(\alpha)}}^*$. Since $E_\alpha=P_\alpha+T_\alpha$ with $P_\alpha$ the orthogonal projector onto $\mathcal K_\alpha$ and $\operatorname{range} T_\alpha\subset \mathcal{T}$ as proved in \Cref{sec:proof sectors}, 
    $E_\alpha\mathcal F_{\Phi^*}E_\alpha=P_\alpha\mathcal F_{\Phi^*}P_\alpha$. Since $\mathcal K_\alpha$ is an enclosure, \cite[Proposition~5.4]{carbone2016irreducible} implies $P_\alpha\mathcal F_{\Phi^*}P_\alpha$ is a set of fixed points of $\Phi^*$.

    Let $x\in \mathcal F_{\Phi^*}$, then $\mathcal T\subset \ker x\cap \ker x^*$, thus $P_\alpha x P_\alpha=E_\alpha^{-\frac12}x E_\alpha^{-\frac 12}$. Thus $\Phi^*(E_\alpha^{-\frac12}x E_\alpha^{-\frac 12})=E_\alpha^{-\frac12}x E_\alpha^{-\frac 12}$. 
    Hence, $x$ is a fixed point of ${\Phi^{(\alpha)}}^*$.

    Assume $x$ is a fixed point of ${\Phi^{(\alpha)}}^*$. Then, by definition of $\Phi^{(\alpha)}$, there exist $y\in \mathcal B(\mathcal H_\alpha)$ such that
    $x=E_\alpha^{\frac12}y E_\alpha^{\frac12}$ and $y$ is a fixed point of $\Phi^*$. Hence, $x\in E_\alpha^{\frac12}\mathcal F_{\Phi^*}E_\alpha^{\frac12}$. Then the equality $E_\alpha^{\frac12}\mathcal F_{\Phi^*}E_\alpha^{\frac12}=P_\alpha\mathcal F_{\Phi^*}P_\alpha=E_\alpha\mathcal F_{\Phi^*}E_\alpha$ yields the proposition.
\end{proof}

Next lemma expresses that each measure $\pp_{\rho,\alpha}$ can be expressed using the deformed instrument $\mathcal J^{(\alpha)}$.
\begin{lemma}\label{lem:law deformed instrument}
    For any $\rho\in \mathcal D(\mathcal H)$ and $\alpha\in\mathcal S$ such that $\tr(E_\alpha\rho)>0$, on $(\Omega,\mathcal F)$ we have
    $$\pp_{\rho,\alpha}({{\mathbf a}})=\tr(\rho^{(\alpha)} \Phi_{a_1}^{(\alpha)}\circ\dotsb\circ\Phi_{a_p}^{(\alpha)}(\id_{\mathcal H_\alpha}))$$
    for any $\mathbf{a}=(a_1,\ldots,a_p)\in \Omega_{\rm fin.}$ with 
    \begin{equation}
        \rho^{(\alpha)}=\frac{E_\alpha^{\frac12}\rho E_\alpha^{\frac12}}{\tr(E_\alpha\rho)}.
        \label{eq:state_i}
    \end{equation}
\end{lemma}
\begin{proof}
   From the explicit expression of $\Phi_a^{(\alpha)}$ and $\rho^{(\alpha)}$ we obtain
   \begin{align*}
       \tr(\rho^{(\alpha)} \Phi_{a_1}^{(\alpha)}\circ\dotsb\circ\Phi_{a_p}^{(\alpha)}(\id_{\mathcal H_\alpha})) &= \frac{\tr(E_\alpha^{\frac12}\rho E_\alpha^{\frac12} E_\alpha^{-\frac12} \Phi_{a_1}\circ\dotsb\circ\Phi_{a_p}(E_\alpha^{\frac12} \id_{\mathcal H_\alpha} E_\alpha^{\frac12} )E_\alpha^{-\frac12}) }{\tr(E_\alpha\rho)} \\
       &= \frac{\tr(\rho \id_{\mathcal{H}_\alpha} \Phi_{a_1}\circ\dotsb\circ\Phi_{a_p}(E_\alpha)\id_{\mathcal{H}_\alpha}) }{\tr(E_\alpha\rho)} \\
       &= \frac{\tr(\rho  \Phi_{a_1}\circ\dotsb\circ\Phi_{a_p}(E_\alpha))}{\tr(E_\alpha\rho)} = \pp_{\rho,\alpha}({\mathbf a})
   \end{align*} 
   where we have used the fact that $\supp \Phi_{a_1}\circ\dotsb\circ\Phi_{a_p}(E_\alpha) \subseteq \mathcal{H}_\alpha$ since $E_\alpha$ is $\Phi$-invariant.
\end{proof}
We will use this expression for the laws $\pp_{\rho,\alpha}$ in our subsequent proofs of the main results.

\section{Exponentially fast selection of a sector: \Cref{thm:generalized QND} and \Cref{cor:exp to subspace} proofs}\label{sec:as}

We introduce a decomposition of $\pp_\rho$ relating it to the partition $\{\Omega_\alpha\}_{\alpha\in \mathcal S}$ and the conditionned measures $\pp_{\rho,\alpha}$.
\begin{lemma}\label{lemma:sectO}
    \label{lem:Omega_alpha proba 1}
    For any $\rho\in \mathcal D(\mathcal H)$ and $\alpha,\beta\in \mathcal S$ such that $\tr(\rho E_\alpha)>0$, $\pp_{\rho,\alpha}(\Omega_\beta)=\delta_{\alpha,\beta}$. Moreover,
    $$\d \pp_{\rho,\alpha}=\frac{\one_{\Omega_\alpha}}{Q_0(\alpha)}\d\pp_\rho,$$
    meaning that for any $A\in \mathcal F$,
    $$\pp_{\rho,\alpha}(A)=\frac{\pp_\rho(A\cap \Omega_\alpha)}{Q_0(\alpha)}.$$
\end{lemma}
\begin{proof} 
    Since $\{E_\alpha\}_{\alpha\in \mathcal S}$ is a POVM, $Q_0(\alpha)=\tr(E_\alpha\rho)$ and $\pp_{\rho}=\sum_{\alpha\in \mathcal S}\tr(E_\alpha \rho)\pp_{\rho,\alpha}$,
    it is sufficient to prove $\pp_{\rho,\alpha}(\Omega_\beta)=\delta_{\alpha,\beta}$. Indeed, then $\pp_{\rho,\alpha}(A)=\pp_{\rho,\alpha}(A\cap \Omega_\alpha)$ and $\pp_{\rho,\beta}(A\cap\Omega_\alpha)=0$ for any $\beta\neq \alpha$. Thus,
    $$\pp_{\rho}(A\cap \Omega_\alpha)=\sum_{\beta\in \mathcal S}Q_0(\beta)\pp_{\rho,\beta}(A\cap \Omega_\alpha)=Q_0(\alpha)\pp_{\rho,\alpha}(A).$$

    Recall that $\Omega_\beta$ is $\theta$-invariant. Hence, \Cref{lem:pp Cesaro convergence,lem:law deformed instrument,prop:deformed instrument} yield,
    $$\pp_{\rho,\alpha}(\Omega_\beta)=\pp_{T_{\infty}(\rho),\alpha}(\Omega_\beta).$$
    Then, \Cref{thm:sectors POVM} Item~(2) yields $\pp_{T_\infty(\rho),\alpha}=\pp_{\alpha}$. Finally, \Cref{thm:sectors ergodic} yields the lemma.
\end{proof}
We now prove the first part of the theorem.
\begin{lemma}\label{lem:conv Q}
    For any $\rho\in \mathcal D(\mathcal H)$, the limits
    $$Q_\infty(\alpha)=\lim_{n\to\infty} Q_n(\alpha)\quad \mbox{and}\quad \widehat{Q}_\infty(\alpha)=\lim_{n\to\infty}\widehat{Q}_n(\alpha)$$
    exist $\pp_\rho$ almost surely and
    $$Q_\infty(\alpha)=\widehat{Q}_\infty(\alpha)=\one_{\Omega_\alpha}, \quad \pp_\rho\as$$
    with $\pp_{\rho}(Q_\infty(\alpha)=1)=\pp_\rho(\Gamma=\alpha)=\pp_\rho(\Omega_\alpha)=Q_0(\alpha)=\tr(E_\alpha\rho)$.
\end{lemma}
\begin{proof}
    By definition,
    $$\left.\d\pp_{\rho,\alpha}\right|_{\mathcal F_n}=\frac{Q_n(\alpha)}{Q_0(\alpha)}\left.\d\pp_\rho\right|_{\mathcal F_n}.$$
    Hence, $Q_n(\alpha)=Q_0(\alpha)\left.\frac{\d\pp_{\rho,\alpha}}{\d\pp_{\rho}}\right|_{\mathcal F_n}$ and $(Q_n(\alpha))_n$ is a bounded martingale. It therefore converges $\pp_\rho$-almost surely and in $L^1(\pp_\rho)$-norm to $Q_\infty(\alpha)$ by Doob's martingale convergence theorem. 

    From the $L^1$ convergence, we deduce that $(Q_n(\alpha))$ is a closed martingale and then for all $n$, we have
    $$Q_n(\alpha)=\mathbb E_{\rho}[Q_\infty(\alpha)\vert\mathcal F_n].$$
    This allows to deduce that
    $$\d\pp_{\rho,\alpha}=\frac{Q_\infty(\alpha)}{Q_0(\alpha)}\d\pp_\rho.$$
    Following \Cref{lem:Omega_alpha proba 1}, $\d\pp_{\rho,\alpha}=\frac{\one_{\Omega_\alpha}}{Q_0(\alpha)}\d\pp_\rho$. Uniqueness of Radon-Nikodym derivative implies $Q_{\infty}(\alpha)=\one_{\Omega_\alpha}$ $\pp_\rho$-almost surely.
    
    Now, we also have $\widehat{Q}_\infty(\alpha):=\lim_{n\to\infty}\widehat{Q}_n(\alpha)=\one_{\Omega_\alpha}$ $\pp_{\hat \rho}$-almost surely by the same argumentation. 
    Since there exist $c>0$ such that $\rho\leq c\hat \rho$, $\pp_\rho \leq c\pp_{\hat \rho}$ by posititivity of the linear extension of $\rho\mapsto \pp_\rho$. Then, the convergence and the equality hold also $\pp_\rho$-almost surely.
    
    It remains to prove that $\pp_\rho(Q_\infty(\alpha)=1)=Q_0(\alpha)$. That is a direct consequence of $\pp_\rho=\sum_{\beta\in \mathcal S}Q_0(\beta)\pp_{\rho,\beta}$ and \Cref{lem:Omega_alpha proba 1}. This concludes the proof of the lemma.
\end{proof}

We now focus on the proof that the convergence is exponentially fast.
The law of outcomes with respect to an invariant initial state is always a convex combination of the probability measures defining the sectors.
\begin{lemma}
    \label{lem:conv combination of sectors}
    Assume $\rho\in \mathcal D_{\Phi^*}$, then
    $$\pp_\rho=\sum_{\alpha\in \mathcal S}\tr(\rho E_\alpha)\pp_\alpha.$$
\end{lemma}
\begin{proof}
    By definition of $\pp_{\rho,\alpha}$, $\pp_\rho=\sum_{\alpha\in \mathcal S}\tr(E_\alpha\rho)\pp_{\rho,\alpha}$.
    Then, \Cref{thm:sectors POVM} Item~(2) yields the lemma.
\end{proof}
We will use the following standard lemma about Radon-Nikodym derivatives. For the reader convenience we provide a short proof.
\begin{lemma}
    \label{lem:strictly positive Randon Nikodym}
    Let $\mu$ and $\nu$ be two probability measures. Assume $\nu\ll \mu$, then $\d\nu/\d\mu>0$, $\nu$-almost surely.
\end{lemma}
\begin{proof}
    Let $A=\{\omega: \d\nu/\d\mu(\omega)=0\}$. Assume $\nu(A)>0$. Then, $\nu(A)=\ee_\mu(\one_A\d\nu/\d\mu)=0$ which is a contradiction. Hence, $\d\nu/\d\mu>0$, $\nu$-almost surely.
\end{proof}
Recall that we denote the chaotic state by $\rho^{\rm ch}=\id_{\mathcal H}/\dim\mathcal H$. We shall use it as a reference state and show some absolute continuity properties on $\pp_{\rho^{\rm ch}}$.
\begin{lemma}
    \label{lem:absolute continuity pp pp1pp2}
    The limit
    $$\lim_{n\to\infty}\frac{\pp_{\rho^{\rm ch}}(\omega_1,\omega_2,\dotsc,\omega_n)}{\pp_{\rho^{\rm ch}}(\omega_2,\omega_3,\dotsc,\omega_n)}$$
    exists and is strictly positive $\pp_{\rho^{\rm ch}}$-almost surely.
\end{lemma}
\begin{proof} 
    Consider $\pp_{\rho_{\rm ch}}$ and $\pp^{(1)}$ be the restriction of $\pp_{\rho_{\rm ch}}$ to $\mathcal F_1$. Then $\mathbb Q=\pp^{(1)}
    \otimes\pp_{\rho_{\rm ch}}$ defines a probability measure over $\Omega$ where $\omega_1$ is independent of $\theta(\omega)$ with 
    $\omega_1\sim \pp^{(1)}$ and $\theta(\omega)\sim\pp_{\rho_{\rm ch}}$. For any $\mathbf{b}\in\Omega_{\rm fin.}$, $\pp_{\rho_{\rm ch}}({\mathbf{b}})=\pp^{(1)}({b_1})\pp_{\rho_{\rm ch}}({\mathbf{b}}|C_{b_1})$. 
    Setting $\tilde{\mathbf{b}}=(b_2,\dotsc,b_n)$ for any 
    $\mathbf{b}=(b_1,\dotsc,b_n)\in \Omega_{\rm fin.}$ and setting $\rho_{b_1}=\Phi_{b_1}^*(\rho_{\rm ch})/\tr(\Phi_{b_1}^*(\rho_{\rm ch}))$, 
    by definition of $\pp_{\rho_{\rm ch}}$,
    $$\pp_{\rho_{\rm ch}}({\mathbf{b}})=\pp^{(1)}({b_1})\pp_{\rho_{b_1}}({\mathbf{\tilde b}}).$$
    Since $\rho_{b_1}\leq \dim\mathcal H\ \rho_{\rm ch}$, it follows from the positivity of $\varrho\mapsto \pp_\varrho$ that,
    $$\pp_{\rho_{\rm ch}}({\mathbf{b}})\leq \dim\mathcal H\ \mathbb Q({\mathbf{b}}).$$
    Hence, since $\mathbf{b}$ was arbitrary, $\pp_{\rho_{\rm ch}}\leq \dim\mathcal H\ \mathbb Q$ and in particular $\pp_{\rho_{\rm ch}}\ll \mathbb Q$. 
    
    The Radon-Nikodym derivative of $\pp_{\rho_{\rm ch}}$ with respect to $\mathbb Q$ restricted to $\mathcal F_n$ is given by:
    $$M_n:=\left.\frac{\d \pp_{\rho_{\rm ch}}}{\d \mathbb Q}\right|_{\mathcal F_n}=\frac{\pp_{\rho_{\rm ch}}({\omega_1,\dotsc,\omega_n})}{\pp_{\rho_{\rm ch}}({\omega_1})\pp_{\rho_{\rm ch}}({\omega_2,\dotsc,\omega_n})}.$$
    As a closed martingale $(M_n)_{n\in \nn}$ converges $\mathbb Q$-almost surely and therefore $\pp_{\rho_{\rm ch}}$-almost surely. In the denominator, $\pp^{(1)}(C_{\omega_1})$ is independent of $n$ and $\pp_{\rho_{\rm ch}}$-almost surely strictly positive. Hence, $(\pp_{\rho_{\rm ch}}(C_{\omega_1})M_n)_{n\in \nn}$ converges $\pp_{\rho_{\rm ch}}$-almost surely and \Cref{lem:strictly positive Randon Nikodym} yields the lemma.
\end{proof}

Before we prove \Cref{thm:generalized QND}, we show two sequences of random variables are sub additive.
\begin{lemma}\label{lem:L subadditive}
    For any $\alpha\in \mathcal S$, let 
    $$L_n(\alpha):\omega\mapsto\ln \pp_\alpha({\omega_1,\dotsc,\omega_n})\qquad\mbox{and}\qquad L_n^{\rm ch}(\alpha):\omega\mapsto\ln \pp_{\rho^{\rm ch},\alpha}({\omega_1,\dotsc,\omega_n}).$$
    Then there exists $C>0$ such that for any $n,m\in \nn$,
    $$L_{n+m}(\alpha)\leq C+ L_n(\alpha)+L_m(\alpha)\circ\theta^n\qquad\mbox{and}\qquad L_{n+m}^{\rm ch}(\alpha)\leq C+ L_n^{\rm ch}(\alpha)+L_m^{\rm ch}(\alpha)\circ\theta^n$$
\end{lemma}
\begin{proof}
Let $\alpha\in\mathcal S$ and $\varrho_\alpha\in\mathcal D(\mathcal H)$ be $\Phi^*$-invariant such that $\supp\varrho_\alpha=\mathcal K_\alpha$. 
In the sequel we repeatedly use the channels $\Phi^{(\alpha)}$ and the fact that Hölder inequality for matrix Schatten norms implies the inequality $\tr(XY)\leq\tr(X)\Vert Y\Vert_\infty$ for any positive semi-definite matrix $X,Y$. 
We also use the operator $P_\alpha$ which stands for the orthogonal projector onto $\mathcal K_\alpha$. Recall that since $\varrho_\alpha$ is ${\Phi^{(\alpha)}}^*$-invariant, positivity implies that for any $\textbf{a}\in \Omega_{\rm fin.}$ 
$${\Phi_{\textbf{a}}^{(\alpha)}}^*(P_\alpha\rho P_\alpha)=P_\alpha{\Phi_{\textbf{a}}^{(\alpha)}}^*(P_\alpha\rho P_\alpha)P_\alpha,$$ 
for all $\rho\in\mathcal D$. Using $P_\alpha\varrho_\alpha P_\alpha=\varrho_\alpha$,
\begin{eqnarray*}
    \mathbb P_{\varrho_\alpha}(C_{\omega_1,\ldots,\omega_{n+m}})&=&\tr(\varrho_\alpha\Phi^{(\alpha)}_{\omega_1}\circ\ldots\Phi^{(\alpha)}_{\omega_n}\circ\Phi^{(\alpha)}_{\omega_{n+1}}\circ\ldots\circ\Phi^{(\alpha)}_{\omega_{n+m}}(\id_{\mathcal H_\alpha}))\\
    &=&\tr({\Phi^{(\alpha)}_{\omega_n}}^*\circ\ldots{\Phi^{(\alpha)}_{\omega_1}}^*(\varrho_\alpha)\Phi^{(\alpha)}_{\omega_{n+1}}\circ\ldots\circ\Phi^{(\alpha)}_{\omega_{n+m}}(\id_{\mathcal H_\alpha}))\\
    &=&\tr( P_\alpha{\Phi^{(\alpha)}_{\omega_n}}^*\circ\ldots{\Phi^{(\alpha)}_{\omega_1}}^*(\varrho_\alpha) P_\alpha\Phi^{(\alpha)}_{\omega_{n+1}}\circ\ldots\circ\Phi^{(\alpha)}_{\omega_{n+m}}(\id_{\mathcal H_\alpha}))\\
    &=&\tr({\Phi^{(\alpha)}_{\omega_n}}^*\circ\ldots{\Phi^{(\alpha)}_{\omega_1}}^*(\varrho_\alpha) P_\alpha\Phi^{(\alpha)}_{\omega_{n+1}}\circ\ldots\circ\Phi^{(\alpha)}_{\omega_{n+m}}(\id_{\mathcal H_\alpha}) P_\alpha)\\
&\leq&\tr(\varrho_\alpha\Phi^{(\alpha)}_{\omega_1}\circ\ldots\Phi^{(\alpha)}_{\omega_n}(\id_{\mathcal H_\alpha}))\Vert P_\alpha \Phi^{(\alpha)}_{\omega_{n+1}}\circ\ldots\circ\Phi^{(\alpha)}_{\omega_{n+m}}(\id_{\mathcal H_\alpha})P_\alpha\Vert_{\infty}
\end{eqnarray*}
Now, there exists $\lambda_\alpha>0$ such that $ P_\alpha\leq\lambda_\alpha\varrho_\alpha$. Indeed, on $P_\alpha\mathcal H$, $\varrho_\alpha$ is faithful and one can choose $1/\lambda_\alpha$ as the minimal eigenvalue of $\varrho_\alpha$ on $\supp \varrho_\alpha$
\begin{eqnarray*}\Vert P_\alpha \Phi^{(\alpha)}_{\omega_{n+1}}\circ\ldots\circ\Phi^{(\alpha)}_{\omega_{n+m}}(\id_{\mathcal H_\alpha})P_\alpha\Vert_{\infty}&\leq& \tr( P_\alpha \Phi^{(\alpha)}_{\omega_{n+1}}\circ\ldots\circ\Phi^{(\alpha)}_{\omega_{n+m}}(\id_{\mathcal H_\alpha}))\\&\leq& \lambda_\alpha\tr(\varrho_\alpha \Phi^{(\alpha)}_{\omega_{n+1}}\circ\ldots\circ\Phi^{(\alpha)}_{\omega_{n+m}}(\id_{\mathcal H_\alpha})).
\end{eqnarray*}
Then,
$$\mathbb P_{\varrho_\alpha}({\omega_1,\ldots,\omega_{n+m}})\leq\lambda_\alpha\mathbb P_{\varrho_\alpha}({\omega_1,\ldots,\omega_{n}})\mathbb P_{\varrho_\alpha}({\omega_{n+1},\ldots,\omega_{n+m}})$$
This way defining $C=\log \lambda_\alpha$, the sub-additivity of $L_n(\alpha)+C$ follows.

Concerning $L_n^{\rm ch}(\alpha)$, 
\begin{align*}
    \tr(E_\alpha)\pp_{\rho_{\rm ch},\alpha}({\omega_1,\ldots,\omega_{n+m}})=& \tr(E_\alpha\Phi^{(\alpha)}_{\omega_1}\circ\ldots\Phi^{(\alpha)}_{\omega_n}\circ\Phi^{(\alpha)}_{\omega_{n+1}}\circ\ldots\circ\Phi^{(\alpha)}_{\omega_{n+m}}(\id_{\mathcal H_\alpha}))\\
        \leq& \tr(E_\alpha\Phi^{(\alpha)}_{\omega_1}\circ\ldots\Phi^{(\alpha)}_{\omega_n}(\id_{\mathcal H_\alpha}))\Vert \Phi^{(\alpha)}_{\omega_{n+1}}\circ\ldots\circ\Phi^{(\alpha)}_{\omega_{n+m}}(\id_{\mathcal H_\alpha})\Vert_{\infty}.
\end{align*}
Since $E_\alpha$ is positive definite on $\mathcal H_\alpha$, there exists $\mu_\alpha>0$, such that $\mu_\alpha E_\alpha\geq \id_{\mathcal H_\alpha}$ and
\begin{align*}
    \Vert \Phi^{(\alpha)}_{\omega_{n+1}}\circ\ldots\circ\Phi^{(\alpha)}_{\omega_{n+m}}(\id_{\mathcal H_\alpha})\Vert_{\infty}\leq & \mu_\alpha \tr(E_\alpha \Phi^{(\alpha)}_{\omega_{n+1}}\circ\ldots\circ\Phi^{(\alpha)}_{\omega_{n+m}}(\id_{\mathcal H_\alpha})).
\end{align*}
It follows,
$$\pp_{\rho_{\rm ch},\alpha}({\omega_1,\ldots,\omega_{n+m}})\leq \mu_\alpha\tr(E_\alpha) \pp_{\rho_{\rm ch},\alpha}({\omega_1,\ldots,\omega_{n}})\pp_{\rho_{\rm ch},\alpha}({\omega_{n+1},\ldots,\omega_{n+m}}).$$
Then, by definition of $\pp_{\rho_{\rm ch},\alpha}$,
$$\pp_{\rho_{\rm ch},\alpha}({\omega_1,\ldots,\omega_{n+m}})\leq C \pp_{\rho_{\rm ch},\alpha}({\omega_1,\ldots,\omega_{n}})\pp_{\rho_{\rm ch},\alpha}({\omega_{n+1},\ldots,\omega_{n+m}})$$
with $C=\tr(E_\alpha)\mu_\alpha$. This inequality between probabilities implies the sub-additivity of $L_n^{\rm ch}(\alpha)+C$. Taking a constant $C$ large enough so that all the sub-additivities hold yields the lemma.
\end{proof}

We are now in position to prove  \Cref{thm:generalized QND}.
\begin{proof}[\Cref{thm:generalized QND} proof]
    The first part of the theorem is proved by \Cref{lem:conv Q}. We focus on the exponential convergence. Let $d=\dim\mathcal H$.
   Since $\rho\leq d\rho^{\rm ch}$, positivity of the linear extension of $\varrho\mapsto \pp_\varrho$ implies $Q_n(\alpha)\leq dQ_0(\alpha)\left.\frac{\d\pp_{\rho^{\rm ch},\alpha}}{\d\pp_{\rho}}\right|_{\mathcal F_n}$. It also implies $\pp_\rho\ll \pp_{\rho^{\rm ch}}$. Thus,
    $$Q_n(\alpha)\leq dQ_0(\alpha)\left.\frac{\d\pp_{\rho^{\rm ch},\alpha}}{\d\pp_{\rho^{\rm ch}}}\right|_{\mathcal F_n}\left.\frac{\d\pp_{\rho^{\rm ch}}}{\d\pp_{\rho}}\right|_{\mathcal F_n}, \quad \pp_\rho\as$$
    Then, again by absolute continuity of $\pp_\rho$ with respect to $\pp_{\rho^{\rm ch}}$,
    $$\lim_{n\to\infty}\left.\frac{\d\pp_{\rho}}{\d\pp_{\rho^{\rm ch}}}\right|_{\mathcal F_n}=\frac{\d\pp_{\rho}}{\d\pp_{\rho^{\rm ch}}}, \quad \pp_\rho\as$$
    and $\frac{\d\pp_{\rho}}{\d\pp_{\rho^{\rm ch}}}>0$ $\pp_\rho$-almost surely by \Cref{lem:strictly positive Randon Nikodym}. Hence,
    $$\limsup_{n\to\infty} \frac1n\ln Q_n(\alpha)\leq \limsup_{n\to\infty}\frac1n\ln \left.\frac{\d\pp_{\rho^{\rm ch},\alpha}}{\d\pp_{\rho^{\rm ch}}}\right|_{\mathcal F_n}, \quad \pp_{\rho}\as$$
    Our goal is to upper-bound the right hand side limit superior $\pp_{\rho^{\rm ch}}$-almost surely. Let,
    $$r_\alpha=-\limsup_{n\to\infty}\frac1n\ln\left.\frac{\d\pp_{\rho^{\rm ch},\alpha}}{\d\pp_{\rho^{\rm ch}}}\right|_{\mathcal F_n}.$$

    By \Cref{lem:strictly positive Randon Nikodym}, since $\pp_\gamma\ll \pp_{\rho^{\rm ch}}$ for any $\gamma\in\mathcal S$, $\pp_\gamma$-almost surely,
    \begin{eqnarray*}\limsup_{n\to\infty}\frac1n \ln\left.\frac{\d\pp_{\rho^{\rm ch},\alpha}}{\d\pp_{\rho^{\rm ch}}}\right|_{\mathcal F_n}&=&\limsup_{n\to\infty}\frac1n \ln\left.\frac{\d\pp_{\rho^{\rm ch},\alpha}}{\d\pp_{\gamma}}\right|_{\mathcal F_n}\left.\frac{\d\pp_{\gamma}}{\d\pp_{\rho^{\rm ch}}}\right|_{\mathcal F_n}\\&=&\limsup_{n\to\infty}\frac1n \ln\left.\frac{\d\pp_{\rho^{\rm ch},\alpha}}{\d\pp_{\gamma}}\right|_{\mathcal F_n}-\frac1n\ln\left.\frac{\d\pp_{\rho^{\rm ch}}}{\d\pp_{\gamma}}\right|_{\mathcal F_n}\\&=&\limsup_{n\to\infty}\frac1n \ln\left.\frac{\d\pp_{\rho^{\rm ch},\alpha}}{\d\pp_{\gamma}}\right|_{\mathcal F_n}.
    \end{eqnarray*}

    By \Cref{lem:L subadditive}, there exists $C>0$ such that $(L_n(\gamma)+C)_{n\in \nn}$ and $(L_n^{\rm ch}(\alpha)+C)_{n\in \nn}$ are both subadditive. The relative entropy  of $\pp_\gamma|_{\mathcal F_n}$ with respect to $\pp_{\rho^{\rm ch},\alpha}|_{\mathcal F_n}$ is
    $$S(\pp_\gamma|_{\mathcal F_n}|\pp_{\rho^{\rm ch},\alpha}|_{\mathcal F_n})=\begin{cases}
    -\ee_\gamma\left(\ln\left.\frac{\d\pp_{\rho^{\rm ch},\alpha}}{\d\pp_{\gamma}}\right|_{\mathcal F_n}\right) & \mbox{ if }\pp_\gamma|_{\mathcal F_n}\ll \pp_{\rho^{\rm ch},\alpha}|_{\mathcal F_n}\\
    \infty &\mbox{ else}
    \end{cases}$$
    and the entropy of $\pp_\gamma|_{\mathcal F_n}$ is
    $$S(\pp_\gamma|_{\mathcal F_n})=-\ee_\gamma(L_n(\gamma)).$$
    Since the entropy is subadditive and $(L_n^{\rm ch}(\alpha)+C)_{n\in \nn}$ is subadditive, using 
    $$S(\pp_\gamma|_{\mathcal F_n}|\pp_{\rho^{\rm ch},\alpha}|_{\mathcal F_n})=-S(\pp_\gamma|_{\mathcal F_n})-\ee_\gamma(L_n^{\rm ch}(\alpha)),$$
    it follows $(S(\pp_\gamma|_{\mathcal F_n}|\pp_{\rho^{\rm ch},\alpha}|_{\mathcal F_n})-C)_{n\in \nn}$ is super additive and Fekete's lemma implies
    $$s(\gamma|\alpha)=\lim_{n\to\infty} \frac1nS(\pp_\gamma|_{\mathcal F_n}|\pp_{\rho^{\rm ch},\alpha}|_{\mathcal F_n})=\sup_{n}\frac{S(\pp_\gamma|_{\mathcal F_n}|\pp_{\rho^{\rm ch},\alpha}|_{\mathcal F_n})-C}{n}.$$
    Hence, $s(\gamma|\alpha)=0$ implies, $\sup_n S(\pp_\gamma|_{\mathcal F_n}|\pp_{\rho^{\rm ch},\alpha}|_{\mathcal F_n})<\infty$. Since the relative entropy is lower semi-continuous, it implies $S(\pp_\gamma|\pp_{\rho^{\rm ch},\alpha})<\infty$. Thus, $\pp_{\rho^{\rm ch},\alpha}\gg \pp_\gamma$ and from \Cref{lem:Omega_alpha proba 1}, $\pp_{\gamma}(\Omega_\alpha)=1$ so that $\alpha=\gamma$. Thus, $s(\gamma|\alpha)= 0$ implies $\alpha=\gamma$. If $\alpha=\gamma$, $\pp_\gamma\leq \dim\mathcal H\ \pp_{\rho_{\rm ch},\gamma}$ implies $\sup_nS(\pp_\gamma|_{\mathcal F_n}|\pp_{\rho^{\rm ch},\gamma}|_{\mathcal F_n})\leq \ln \dim\mathcal H$. Hence, $s(\gamma|\gamma)=0$. It follows $s(\alpha|\gamma)\geq 0$ with equality if and only if $\alpha=\gamma$.
    
   Now, by Kingman's subadditive ergodic theorem,
    $$\lim_{n\to \infty} \frac1n\ln\left.\frac{\d\pp_{\rho^{\rm ch},\alpha}}{\d\pp_{\gamma}}\right|_{\mathcal F_n}=\lim_{n\to\infty}\frac1n(L_n^{\rm ch}(\alpha)-L_n(\gamma))=-s(\gamma|\alpha),\quad \pp_\gamma\as$$

    Replacing $\gamma$ with $\Gamma$ which are equal $\pp_\gamma$-almost surely,
    $$r_\alpha=s(\Gamma|\alpha),\quad\pp_\gamma\as$$
    Hence, $r_\alpha=s(\Gamma|\alpha)$ $\pp_{\rho_\infty}$-almost surely for any $\rho_\infty\in \mathcal D_{\Phi^*}$ thanks to \Cref{lem:conv combination of sectors}.

    It remains to prove $s(\Gamma|\alpha)$ upper bounds $r_\alpha$ $\pp_{\rho^{\rm ch}}$-almost surely. We write the Radon-Nikodym derivative explicitly:
    $$\left.\frac{\d\pp_{\rho^{\rm ch},\alpha}}{\d\pp_{\rho^{\rm ch}}}\right|_{\mathcal F_n}=\frac{\tr[\Phi_{\omega_1}\circ\dotsb\circ\Phi_{\omega_n}(E_\alpha)]}{\tr[\Phi_{\omega_1}\circ\dotsb\circ\Phi_{\omega_n}(\id)]}.$$
    Since $\|\Phi_{a}^*(\id)\|_\infty\leq d$ for any $a\in\mathcal A$, using Hölder's inequality $\tr(AB)\leq \|A\|_\infty\tr(B)$ for positive semi-definite $A$ and $B$ and $\|\Phi_{\omega_1}^*(\id_{\mathcal H})\|_\infty\leq d$, the numerator is such that,
    $$\tr[\Phi_{\omega_1}\circ\dotsb\circ\Phi_{\omega_n}(E_\alpha)]\leq d\ \tr[\Phi_{\omega_2}\circ\dotsb\circ\Phi_{\omega_n}(E_\alpha)].$$
    Using \Cref{lem:absolute continuity pp pp1pp2}, the denominator verifies
    $$\lim_{n\to\infty}\frac{\tr[\Phi_{\omega_1}\circ\dotsb\circ\Phi_{\omega_n}(\id)]}{\tr[\Phi_{\omega_2}\circ\dotsb\circ\Phi_{\omega_n}(\id)]}=\lim_{n\to\infty}\frac{\pp_{\rho^{\rm ch}}({\omega_1,\omega_2,\dotsc,\omega_n})}{\pp_{\rho^{\rm ch}}({\omega_2,\omega_3,\dotsc,\omega_n})}>0,\quad \pp_{\rho^{\rm ch}}\as$$
    Thus,
    $$r_\alpha\geq r_\alpha\circ\theta, \quad \pp_{\rho^{\rm ch}}\as$$
    Since for any $\rho\in \mathcal D(\mathcal H)$, $\pp_\rho\ll \pp_{\rho^{\rm ch}}$, this inequality also holds $\pp_\rho$-almost surely for any $\rho\in \mathcal D(\mathcal H)$.

    Then, since $r_\alpha\circ\theta\leq r_\alpha$ and $s(\Gamma|\alpha)\circ\theta=s(\Gamma|\alpha)$, for any $\rho\in \mathcal D(\mathcal H)$, $\pp_{\rho}(r_\alpha\geq s(\Gamma|\alpha))\geq \pp_{\rho}(r_\alpha\circ\theta\geq s(\Gamma|\alpha))=\pp_{\rho}\circ\theta^{-1}(r_\alpha\geq s(\Gamma|\alpha))=\pp_{\Phi^*(\rho)}(r_\alpha\geq s(\Gamma|\alpha))$. Repeating this procedure, \Cref{lem:pp Cesaro convergence} and affinity of $\varrho\mapsto \pp_\varrho$ imply,
    $$\pp_{\rho^{\rm ch}}(r_\alpha\geq s(\Gamma|\alpha))\geq\frac1n\sum_{k=1}^n\pp_{{\Phi^*}^k(\rho^{\rm ch})}(r_\alpha\geq s(\Gamma|\alpha))\xrightarrow[n\to\infty]{} \pp_{T_\infty(\rho^{\rm ch})}(r_\alpha\geq s(\Gamma|\alpha))=1.$$
    Thus $\pp_{\rho^{\rm ch}}(r_\alpha\geq s(\Gamma|\alpha))=1$. 
    In order to conclude for the filter, we use the result that we have just proved, namely
\[
\limsup_{n\to\infty} \frac{1}{n} \ln \hat{Q}_n(\alpha) \leq -s(\Gamma\vert\alpha),
\]
$\pp_{\hat{\rho}}$ almost surely. Since $\pp_\rho \ll \pp_{\hat{\rho}}$, we finally deduce that
\[
\limsup_{n\to\infty} \frac{1}{n} \ln \hat{Q}_n(\alpha) \leq -s(\Gamma\vert\alpha),
\]
$\pp_{\rho}$ almost surely.
\end{proof}

\begin{proof}[Proof of \Cref{cor:exp to subspace}]
    Since for any $\alpha\in \mathcal S$, $P_\alpha\leq E_\alpha$, the monotonicity of the logarithmic function implies the bounds on the rate of convergence.  We only need to prove $\tr(P_\Gamma\rho_n)\xrightarrow[n\to\infty]{\mbox{a.s.}}1$.

    \Cref{thm:generalized QND} implies $\tr(E_\Gamma\rho_n)\xrightarrow[n\to\infty]{\mbox{a.s.}}1$. Following the proof of \Cref{thm:sectors POVM}, $E_\Gamma=P_\Gamma+ P_{\mathcal T}E_\Gamma P_{\mathcal T}$ with $P_{\mathcal T}E_\Gamma P_{\mathcal T}\leq P_{\mathcal T}$ where $P_{\mathcal T}$ is the orthogonal projection onto the transient subspace $\mathcal T$. By definition, $\mathbb E_\rho(\tr(P_{\mathcal T}\rho_n))=\tr(P_{\mathcal T}{\Phi^*}^n(\rho))\xrightarrow[n\to\infty]{}0$. It thus remains to prove $(\tr(P_{\mathcal T}\rho_n))_n$ converges almost surely. We are inspired by \cite[Theorem~1.1]{BenoistTicco}. Since $\Phi^*(\mathcal B(\mathcal T^\perp))\subset \mathcal B(\mathcal T^\perp)$, the positivity of $\Phi$ implies $\Phi(\mathcal B(\mathcal T))\subset \mathcal B(\mathcal T)$. Since $\Phi(\id)=\id$, the positivity of $\Phi$ implies $\Phi(P_{\mathcal T})\leq \id$, but $\Phi(P_{\mathcal T})\in \mathcal B(\mathcal T)$, so $\Phi(P_{\mathcal T})\leq P_{\mathcal T}$. It follows $(\tr(P_{\mathcal T}\rho_n))_n$ is a positive super-martingale. Therefore, it converges almost surely and the corollary is proved.
\end{proof}
Note that the quantity \( s(\cdot \mid \cdot) \) defined with \( \pp_\rho^{ch} \) does not depend on the choice of the chaotic state. Indeed, this quantity remains the same for any other faithful state \( \rho \) since there exist a constant $C>0$ such that  
\(
C^{-1}\rho \leq \rho^{ch} \leq C\rho,
\) 
and this constant vanishes at the logarithmic scale when divided by \( n \).

Note that in the QND model (see the example of \Cref{sec:xpl qnd}), the relative entropy has an explicit form \cite{BenrardBauer, BenoistAHP}. This is a consequence of the fact that for QND models $\pp_\rho$ is always the law of a mixture of independent and identically distributed random variables. One can then apply the law of large numbers. Here the situation is quite different and the use of Kingman's theorem ensures the existence of the involved quantities but there is little hope to obtain explicit formulas.

\section{Exponential convergence in mean: \Cref{thm:meanconv} Proof}\label{sec:mean}

The first result of this section is a generalization of \cite[Theorem~3.3]{amini2023exponential}. In \cite{amini2023exponential} a crucial hypothesis was that $\mathcal T=\{0\}$ which eases the computations. Indeed, in that case $\Phi$ is block diagonal and each operator $E_\alpha$ is proportional to an orthogonal projector onto $\mathcal K_\alpha$. Here we use Lemma \ref{lem:law deformed instrument} which overcomes this hypothesis. Furthermore our context allows for considering imperfect measurements and equivalent invariant states in the sense of the sector definition. This was not addressed in \cite{amini2023exponential}.

\begin{proposition}\label{Prop:identif}
    There exists an integer N such that for all $\alpha,\beta\in \mathcal S$ such that $\alpha\neq\beta$ and for all $\varrho,\rho\in \mathcal D(\mathcal H)$ such that $\tr(E_\alpha\varrho)>0$ and $\tr(E_\beta\rho)>0$, and for all $n\geq N$,  there exists a word $\textbf{b}\in\mathcal A^n$  
$$\mathbb P_{\varrho,\alpha}({\textbf{b}}) \neq \mathbb P_{\rho,\beta}({\textbf{b}}).$$
\end{proposition}

\begin{proof}
Some elements of this proof are similar to ones developed in \cite{amini2023exponential}.

First note that if $\textbf{b}\in\mathcal A^k$ is such that
$\mathbb P_{\varrho,\alpha}({\textbf{b}}) \neq \mathbb P_{\rho,\beta}({\textbf{b}})$ then for any $k'\geq k$ there exists $\textbf{b}'\in\mathcal A^{k'}$ such that $\mathbb P_{\varrho,\alpha}({\textbf{b}'}) \neq \mathbb P_{\rho,\beta}({\textbf{b}'}).$ Indeed for $k'\geq k$,  
$$\sum_{\textbf{a}\in\mathcal A^{k'-k}}\mathbb P_{\varrho,\alpha}({\textbf{b}\textbf{a}})=\mathbb P_{\varrho,\alpha}({\textbf{b}})\neq\mathbb P_{\rho,\beta}({\textbf{b}})=\sum_{\textbf{a}\in\mathcal A^{k'-k}}\mathbb P_{\rho,\beta}({\textbf{b}\textbf{a}})$$
and therefore we cannot have $P_{\varrho,\alpha}({\textbf{b}\textbf{a}})=P_{\rho,\beta}({\textbf{b}\textbf{a}})$ for all $\textbf{a}\in\mathcal A^{k'-k}.$ 

If follows that, since $\mathcal S$ is finite, there exists $n_0\in \mathbb N$ such that for any $\alpha,\beta\in \mathcal S$ such that $\alpha\neq \beta$, there exist $\textbf{b}\in\mathcal A^{n_0}$ such that
$$\mathbb P_{\alpha}({\textbf{b}})\neq\mathbb P_{\beta}({\textbf{b}}).$$
Then, let
$$e_{\alpha,\beta}=\vert \mathbb P_{\alpha}({\textbf{b}})-\mathbb P_{\beta}({\textbf{b}})\vert>0$$

Using that $\sigma\mapsto \tr(\sigma E_\gamma)\pp_{\sigma,\gamma}$ is affine and $E_\gamma$ is $\Phi$-invariant, repeating the proof of \Cref{lem:pp Cesaro convergence}, for any $\gamma\in \mathcal S$,
$$\lim_{n\to\infty}\sup_{\substack{\sigma\in \mathcal D(\mathcal H)\\ \tr(\sigma E_\gamma)>0}}\sup_{A\in \mathcal F}\left|\frac1n\sum_{k=1}^n \pp_{\sigma,\gamma}\circ\theta^{-k}(A)-\pp_{\gamma}(A)\right|=0.$$
Indeed, $\pp_{T_\infty(\sigma),\gamma}=\pp_\gamma$ since, by \Cref{lem:law deformed instrument}, $\pp_{T_\infty(\sigma),\gamma}=\sum_{i\in \gamma}\frac{\tr(\sigma E_\alpha)}{\tr(\sigma E_\gamma)}\pp_{\varrho_\alpha}=\pp_\gamma$.

Since $\mathcal S$ is finite, it follows there exists $M$ large enough such that for any $\alpha\neq \beta$ and $\varrho,\rho\in \mathcal D(\mathcal H)$ such that $\tr(\varrho E_\alpha)\tr(\rho E_\beta)>0$,
$$\left|\frac1M\sum_{k=1}^{M}\pp_{\varrho,\alpha}\circ\theta^{-k}({\mathbf b})-\pp_\alpha({\mathbf b})\right|+\left|\frac1M\sum_{k=1}^{M}\pp_{\rho,\beta}\circ\theta^{-k}({\mathbf b})-\pp_\beta({\mathbf b})\right|\leq \frac{e_{\alpha,\beta}}{2}.$$

Then, triangular inequality applied twice implies,
$$\frac1M\sum_{k=1}^{M}\left|\pp_{\varrho,\alpha}\circ\theta^{-k}({\mathbf b})-\pp_{\rho,\beta}\circ\theta^{-k}({\mathbf b})\right|\geq \frac{e_{\alpha,\beta}}{2}.$$
Hence, there exists $k\in\{1,\dotsc,M\}$ such that
$$\left|\pp_{\varrho,\alpha}\circ\theta^{-k}({\mathbf b})-\pp_{\rho,\beta}\circ\theta^{-k}({\mathbf b})\right|\geq \frac{e_{\alpha,\beta}}{2}.$$

Using that $\pp_{\sigma,\gamma}\circ\theta^{-k}({\mathbf b})=\sum_{\mathbf{a}\in \mathcal A^k}\pp_{\sigma,\gamma}({\mathbf{ab}})$ for any $\gamma\in \mathcal S$ and $\sigma\in \mathcal D(\mathcal H)$ such that $\tr(\sigma E_\gamma)>0$, there exist $\mathbf{a}\in \mathcal A^k$ such that,
$$|\pp_{\varrho,\alpha}({\mathbf{ab}}) - \pp_{\rho,\beta}({\mathbf{ab}})|>0.$$

That implies that for any $n\geq n_0+M$, there exists $\mathbf{b}\in \mathcal A^n$ such that
$$\pp_{\varrho,\alpha}({\mathbf{b}}) \neq \pp_{\rho,\beta}({\mathbf{b}}).$$

Since $M$ and $n_0$ are independent of $\alpha, \beta, \varrho$ and $\rho$, setting $N=M+n_0$ yields the proposition.
\end{proof}

Last proposition statement is the essential step of the proofs in \cite{amini2023exponential}. Indeed, all the exponential estimates in \cite{amini2023exponential} rely on the fact that identifiability with regards to the invariant states can be transferred to all the states of the related minimal invariant subspace. Here we improve this result by allowing for invariant states producing identical outcomes laws, non full support invariant states and non perfect measurements.

For the sake of completeness we recall now the results needed to obtain the exponential selection of sectors in mean. The following lemma defines a constant related to the rate of convergence. To formulate it we introduce
$$\mathcal D_{\alpha,\beta}=\{\rho\in\mathcal D(\mathcal H) : \tr(E_\alpha\rho)>0,\tr(E_\beta\rho)>0\}.$$

\begin{lemma}\label{lem:kappa}
    Let $N$ be the integer of Proposition \ref{Prop:identif}. Then,
    $$\kappa^N=\sup_{\alpha\neq\beta}\sup_{\rho\in\mathcal D_{\alpha,\beta}}\sum_{\textbf{a}\in\mathcal A^N}\sqrt{\mathbb P_{\rho,\alpha}({\textbf{a}})\mathbb P_{\rho,\beta}({\textbf{a}})}<1$$
\end{lemma}
\begin{proof}
    Let $\alpha,\beta\in \mathcal S$ be such that $\alpha\neq \beta$. Let $\rho\in \mathcal D_{\alpha,\beta}$. Then, Cauchy-Schwartz inequality implies
    $$\kappa_{\alpha,\beta}(\rho)=\sum_{\mathbf{a}\in \mathcal A^N}\sqrt{\mathbb P_{\rho,\alpha}({\textbf{a}})\mathbb P_{\rho,\beta}({\textbf{a}})}\leq \sqrt{\sum_{\mathbf{a}\in \mathcal A^N}\mathbb P_{\rho,\alpha}({\textbf{a}})}\sqrt{\sum_{\mathbf{a}\in \mathcal A^N}\mathbb P_{\rho,\beta}({\textbf{a}})}=1.$$
    Assume equality holds, then the vectors $(\sqrt{\mathbb P_{\rho,\beta}({\textbf{a}})})_{\mathbf{a}\in \mathcal A^N}$ and $(\sqrt{\mathbb P_{\rho,\beta}({\textbf{a}})})_{\mathbf{a}\in \mathcal A^N}$ are colinear. Since they both have non-negative entries and they both have $\ell^2$ norm $1$, they are equal. That contradicts \Cref{Prop:identif}. Hence, $\kappa_{\alpha,\beta}(\rho)<1$. Since $\mathcal S$ is finite and $\mathcal D(\mathcal H_\gamma)$ with $\mathcal H_\gamma=E_\gamma\mathcal H$ is compact for any $\gamma\in \mathcal S$, $\kappa^N=\sup_{\alpha\neq\beta}\sup_{\rho\in\mathcal D_{\alpha,\beta}}\kappa_{\alpha,\beta}(\rho)<1$ and the lemma is proved.
\end{proof}
We turn to the proof of exponential convergence using Lyapunov function
$$W(\rho)=\sum_{\alpha\neq\beta}\sqrt{\tr(E_\alpha\rho)\tr(E_\beta\rho)},$$
which in terms of the quantum trajectories, reads as
$$W(\rho_n)=\sum_{\alpha\neq\beta}\sqrt{Q_n(\alpha)Q_n(\beta)}.$$
\begin{proof}[\Cref{thm:meanconv} proof]
The proof is a consequence of Proposition \ref{Prop:identif} and Lemma \ref{lem:kappa}. It is similar to the one of \cite[Theorem~3.3]{amini2023exponential}. For any $k\in \nn$, direct computation leads to
\begin{align}\label{eq:decreasing W}
    \begin{split}
    \mathbb E[W(\rho_{k+1})\vert\rho_k]=&\ \sum_{\alpha \neq \beta}   \mathbf 1_{\rho_k\in\mathcal D_{\alpha,\beta}}\sqrt{\tr{E_\alpha \rho_k} \tr{E_\beta \rho_k} } \sum_{a\in\mathcal A}\sqrt{\mathbb P_{\rho_k,\alpha}(C_a)\mathbb P_{\rho_k,\beta}(C_a)}  \\
    \leq&\ W(\rho_k)
    \end{split}
\end{align}
where we used that Cauchy-Schwartz inequality implies $\sum_{a\in\mathcal A}\sqrt{\mathbb P_{\rho_k,\alpha}(C_a)\mathbb P_{\rho_k,\beta}(C_a)}\leq 1$. That implies $\mathbb E(W(\rho_n))$ is non increasing in $n$.

Moreover, \Cref{lem:kappa} implies
\begin{eqnarray*}
  \mathbb E[W(\rho_{k+N})\vert\rho_k]&=&\frac{1}{2} \sum_{\alpha \neq \beta}   \mathbf 1_{\rho_k\in\mathcal D_{\alpha,\beta}}\sqrt{\tr{E_\alpha \rho_k} \tr{E_\beta \rho_k} } \sum_{\mathbf{a}\in\mathcal A^N}\sqrt{\mathbb P_{\rho_k,\alpha}(C_\textbf{a})\mathbb P_{\rho_k,\beta}(C_\textbf{a}})  \\
  &\leq& \kappa^N W(\rho_k),
\end{eqnarray*}
Now for $n\in\mathbb N$, let $q=\lfloor n/N\rfloor$ and $l=n-qN$. Then,
$$\mathbb E[W(\rho_n)]=\mathbb E[W(\rho_{qN+l})]\leq\kappa^{qN}\mathbb E[W(\rho_l)]\leq \kappa^{-(N-1)}W(\rho)\,\kappa^n$$
where we used \Cref{eq:decreasing W} $l$ times to prove $\ee(W(\rho_l))\leq W(\rho)$. Then fixing $\tau=\kappa^{-(N-1)}$ yields the theorem for the true trajectory. Note that if $N=1$, $\tau=1$.

For the filter, for $\hat \rho\in \mathcal D(\mathcal H)$ positive definite, $\|\hat \rho^{-\frac12}\rho \hat \rho^{-\frac12}\|_\infty=\min\{c\geq 0: \rho\leq c\hat \rho\}$. Then, the positivity of $\varrho\mapsto \pp_\varrho$ implies $\pp_\rho\leq \|\hat \rho^{-\frac12}\rho \hat \rho^{-\frac12}\|_\infty\ \pp_{\hat \rho}$. 
Hence,
$$\mathbb E_{\rho}[W(\hat\rho_n)]=\mathbb E_{\hat\rho}\left[W(\hat\rho_n)\frac{\d\pp_\rho}{\d\pp_{\hat \rho}}\right]\leq \|\hat \rho^{-\frac12}\rho \hat \rho^{-\frac12}\|_\infty\ \mathbb E_{\hat\rho}[W(\hat\rho_n)]$$
and the bound for the true trajectory yields the theorem.
\end{proof}

\bigskip
\noindent \textbf{Acknowledgements.} The results presented in this article were obtained during an interneship of L. G. at Institut de Mathématiques de Toulouse. Our presentation of sectors evolved significantly after some useful criticism by anonymous referees and we are grateful for their guidance. During the completion of this work, the authors were supported by the ANR project ``ESQuisses'', grant number ANR-20-CE47-0014-01, the ANR project ``Quantum Trajectories'' grant number ANR-20-CE40-0024-01 and the program ``Investissements d'Avenir'' ANR-11-LABX-0040 of the French National Research Agency. C. P. was also supported by the ANR project Q-COAST ANR-19-CE48-0003.

\bigskip
\noindent\textbf{Conflicts of interests.} The authors declare no conflicts of interests with respect to the present article.

\end{document}